\documentclass[final,5p,twocolumn,sort&compress]{elsarticle}

\usepackage{amsmath,amssymb,amsfonts,mathtools,tensor}
\usepackage{graphicx}

\newcommand{\m}[1]{\boldsymbol{#1}} 
\DeclareMathOperator*{\argmin}{arg\,min}

\newcommand{\mat}[1]{\begin{bmatrix*}[r]#1
\end{bmatrix*}}
\newcommand{\matc}[1]{\begin{bmatrix*}[c]#1
\end{bmatrix*}}
\newcommand{\matl}[1]{\begin{bmatrix*}[l]#1
\end{bmatrix*}}
\usepackage[utf8]{inputenc}

\usepackage{xcolor}
\newcommand\tdd[1]{} 
\newcommand\tdquelled[1]{} 
\newcommand{\nhphantom}[1]{\sbox0{#1}\hspace{-\the\wd0}}

\usepackage{siunitx}

\usepackage{booktabs}
\usepackage{tabu}
\usepackage{multirow}

\usepackage{pgfplots}
\usepackage{varwidth}
\usetikzlibrary{positioning}
\usetikzlibrary{arrows}
\usetikzlibrary{calc} 

\usepackage{amsfonts,amsthm,bm, amssymb}
\usepackage{mathtools}
\usepackage{amsmath}
\usepackage{bm}
\usepackage{mathtools}

\usepackage{verbatim}
\usepackage{pgfplots}
\usetikzlibrary{spy}

\DeclareMathOperator{\rank}{rank}

\usepackage{algorithm} 
\usepackage{algpseudocode} 
\algnewcommand{\Initialize}[1]{%
	\State \textbf{initialize} {\raggedright #1}
}
\makeatletter
\algnewcommand{\Statey}[1]{\Statex \hskip\ALG@thistlm #1}
\makeatother
\algdef{SE}[DOWHILE]{Do}{doWhile}{\algorithmicdo}[1]{\algorithmicwhile\ #1}%

\usepackage{pgfplots}

\newcommand{\norm}[1]{\left\lVert #1 \right\rVert}
\newcommand{\abs}[1]{\left\lvert #1 \right\rvert}

\newdefinition{prob}{Problem}
\newdefinition{assumption}{Assumption}
\newdefinition{definition}{Definition}
\newdefinition{note}{Note}
\newdefinition{remark}{Remark}
\newtheorem{theorem}{Theorem}
\newtheorem{lemma}{Lemma}
\newtheorem{corollary}{Corollary}

\journal{Neurocomputing}

\begin{document}

\begin{frontmatter}



\title{Adaptive Optimal Control for Reference Tracking Independent of Exo-System Dynamics}


\author[1]{Florian K\"opf\corref{cor1}} 
\ead{florian.koepf@kit.edu}
\author[2]{Johannes Westermann} 
\author[1]{Michael Flad}
\author[1]{S\"oren Hohmann}
\cortext[cor1]{Corresponding author}

\address[1]{Institute of Control Systems, Karlsruhe Institute of Technology (KIT), 76131~Karlsruhe, Germany}
\address[2]{Intelligent Sensor-Actuator-Systems Laboratory, Institute for Anthropomatics and Robotics, Karlsruhe Institute of Technology, Karlsruhe, Germany}

\begin{abstract}
Model-free control based on the idea of Reinforcement Learning is a promising approach that has recently gained extensive attention. However, Reinforcement-Learning-based control methods solely focus on the regulation problem or learn to track a reference that is generated by a time-invariant exo-system. In the latter case, controllers are only able to track the time-invariant reference dynamics which they have been trained on and need to be re-trained each time the reference dynamics change. Consequently, these methods fail in a number of applications which obviously rely on a trajectory not being generated by an exo-system. One prominent example is autonomous driving. This paper provides for the first time an adaptive optimal control method capable to track reference trajectories not being generated by a time-invariant exo-system. The main innovation is a novel Q-function that directly incorporates a given reference trajectory on a moving horizon. This new Q-function exhibits a particular structure which allows the design of an efficient, iterative, provably convergent Reinforcement Learning algorithm that enables optimal tracking. Two real-world examples demonstrate the effectiveness of our new method. 
\end{abstract}



\begin{keyword}
Adaptive Dynamic Programming \sep Optimal Tracking \sep Reinforcement Learning \sep Learning Systems \sep Machine Learning \sep Artificial Intelligence \sep Optimal Control


\end{keyword}

\end{frontmatter}


\section{Introduction}
Reinforcement Learning (RL) has recently gained extensive attention as a model-free adaptive optimal control method which learns optimal behavior from interaction with the environment and observations of resulting states and rewards (see e.g. \cite{Sutton.2018, Wang.2017, Luo.2017} and references therein).

The existing control-oriented approaches can be classified in two groups: Some methods focus on the regulation problem, i.e. steering the system towards an equilibrium point (typically~$\m{0}$), e.g. \cite{Kamalapurkar.2016, Luo.2017, Vamvoudakis.2017, Wang.2016, Wei.2017}. The second group is devoted to a specific tracking case, where the reference $\m{r}_k$ to be tracked is generated by an exo-system, i.e. it is assumed that there exists $\m{r}_{k+1}=\m{f}_{\text{ref}}(\m{r}_k)$.

It is noteworthy that a controller that has learned to solve the regulation problem is not directly applicable to the tracking case as the reference trajectory influences the associated rewards that the RL agent is facing and therefore affects the value function which has to be learned.

In order to extend learning-based controllers to the tracking case, various RL-based tracking controllers belonging to the second group of approaches have been proposed recently.
For unknown internal system dynamics but a known control coefficient matrix, adaptive tracking controllers approximating the value function are proposed for discrete-time \cite{Dierks.2009} and continuous-time \cite{Modares.2014, Zhang.2017} systems. Notable critic-only Q-learning methods for completely unknown dynamics are proposed in \cite{Luo.2016} and \cite{Kiumarsi.2014}, where \cite{Kiumarsi.2014} focuses on the linear-quadratic setting and \cite{Luo.2016} allows nonlinear system dynamics. Other works combine system identification and adaptive schemes including past information \cite{Lv.2017}, use filter-based goal representation Heuristic Dynamic Programming for tracking \cite{Mu.2017b}, focus on learning a tracking controller from input-output data rather than assuming full state information \cite{Kiumarsi.2015b}, consider the nonzero-sum game case \cite{Kopf.2018}, utilize the exo-system of a prescribed robot impedance model in order to learn a model-following behavior for assistive human-robot interaction \cite{Modares.2016}, focus on systems with matched uncertainties \cite{Mu.2017} or consider tracking on an infinite horizon with unbounded cost~\cite{Bernhard.2018}.

All of these methods \cite{Dierks.2009, Kiumarsi.2014, Kiumarsi.2015b, Kopf.2018, Luo.2016, Lv.2017, Modares.2014, Modares.2016, Mu.2017, Mu.2017b, Zhang.2017,Bernhard.2018} consider the case where the reference trajectory is generated by a time-invariant exo-system $\m{f}_{\text{ref}}$, i.e. rely on the assumption that the reference $\m{r}_k$ follows $\m{r}_{k+1}=\m{f}_{\text{ref}}(\m{r}_k)$ (respectively $\dot{\m{r}}(t)=\m{f}_{\text{ref}}(\m{r}(t))$ in the continuous-time case).

However, for various applications, this assumption does not hold (e.g. vehicles that follow a road, robots in sophisticated human-machine collaboration or specific time-varying sequences in process engineering). Consequently, present exo-system-based methods are not suited for these types of applications. This is because the learned parameters (i.e. Q-function or value function approximations and associated control laws) are corresponding to the system dynamics and reference dynamics $\m{f}_{\text{ref}}$ and thus need to be re-learned as soon as any other reference $\tilde{\m{f}}_{\text{ref}}\neq\m{f}_{\text{ref}}$ should be followed.
A method that tries to cope with this challenge is the multiple-model approach in \cite{Kiumarsi.2015}, where an adaptive self-organizing map detects changes and switches between various learned models. However, in this approach new sub-models need to be trained whenever the reference dynamics $\m{f}_{\text{ref}}$ changes. 

%

An RL method that tracks reference trajectories not necessarily resulting from a time-invariant exo-system $\m{f}_{\text{ref}}$ (or multiple switched models \cite{Kiumarsi.2015}) has not yet been proposed. Therefore, in contrast to \cite{Dierks.2009, Kiumarsi.2014, Kiumarsi.2015b, Kopf.2018, Luo.2016, Lv.2017, Modares.2014, Modares.2016, Mu.2017, Mu.2017b, Zhang.2017,Bernhard.2018,Kiumarsi.2015}, our work provides for the first time an RL method allowing to track an arbitrary reference trajectory on a moving horizon.

Our general idea is to explicitly incorporate the reference trajectory on a moving horizon in our new Q-function that is learned without requiring the system dynamics. Due to considering not only the current tracking error but rather the reference trajectory, this Q-function yields a controller which does not only achieve reactive but predictive behavior.




\tdd{off-policy learning!? Noch erwähnen, dass Verfahren ähnlich wie tabular Q-learning ein off-policy-Verfahren ist}
\tdd{allg. zu Iteration?}

In particular, we provide:
\begin{itemize}
	\item A novel moving horizon tracking Q-function whose minimizing control also minimizes the tracking costs.
	\item Derivation and proof of the analytical solution of the novel Q-function given the cost function and the system dynamics for the linear-quadratic tracking case.
	\item Convergence proofs for our learning algorithm that learns optimal tracking from data without knowledge of the system dynamics. Thus, the estimated Q-function parameters as well as the associated control law converge to the optimal solution.

\end{itemize}

The rest of this paper is organized as follows. In Section~\ref{problem_def}, we define the optimal tracking problem with unknown system dynamics. The novel Q-function for arbitrary reference tracking is defined and analyzed in Section~\ref{qfunc_section}. In Section~\ref{section_learning}, we introduce our learning algorithm that is based on the previously defined Q-function and provide convergence proofs. Simulation results and a comparison of our method with an adaptive tracking controller that is trained on a reference generated by an exo-system $\m{f}_{\text{ref}}$ are presented and discussed in Section~\ref{sectionsimulation} before the paper is concluded in Section~\ref{conclusion}.

\section{Problem Definition}\label{problem_def}
\tdd{Vergleichbar mit Masterarbeit. ABER: hier bitte mit S = Q, damit letzter Term nicht rausfällt und auch wirklich Nh auftaucht (dadurch verschiebt sich vermutlich einfach nur der Index? Wird damit auch Nh = 1 zulässig?)-> ja!} 

Consider a discrete-time linear system
\begin{align}
\m{x}_{k+1} = \m{A}\m{x}_{k} + \m{B}\m{u}_{k}, \label{eq:System}
\end{align}  
where $k \in \mathbb{N}_0$ is the discrete time step, $\bm{x}_{k} \in \mathbb{R}^{n}$ the state vector and $\bm{u}_k \in \mathbb{R}^{m}$ the input vector. Both the system matrix $\m{A}\in\mathbb{R}^{n\times n}$ and the input matrix $\m{B}\in\mathbb{R}^{n\times m}$ are unknown, albeit the system $(\m{A}, \m{B})$ is assumed to be controllable.

Assume that at each time step k, an arbitrary reference $\tilde{\m{r}}_i^{(k)}\in\mathbb{R}^n$ is given on a moving horizon of length $N$, where $i$ is the time index starting at $k$. Beyond the horizon $N$, let the reference be $\m{0}$. Thus,
\begin{align}\label{eq:ref}
\m{r}_i^{(k)}=\begin{cases}\tilde{\m{r}}_i^{(k)}, &\text{for } i=k,\dots,k+N \\ \m{0}, &\text{for } i=k+N+1,\dots,\infty\end{cases}
\end{align}
follows. 

Our aim is to learn to track $\m{r}_i^{(k)}$ \eqref{eq:ref} optimally w.r.t. the quadratic cost
\begin{align}
J_k &= \sum \limits_{i=k}^{\infty}\gamma^{i-k}\underbrace{\frac{1}{2}\left(\m{e}_i^\intercal\m{Q}\m{e}_i + \m{u}_i^\intercal\m{R}\m{u}_i\right)}_{c\left(\m{x}_i,\m{u}_i, \m{r}_i\right)=c_i} \label{eq:Kostena}
\end{align}
that has to be minimized. Note that we write $\m{r}_i=\m{r}_i^{(k)}$ for brevity of notation in the following as $J_k$ indicates that the reference $\m{r}_i^{(k)}$ starting at $k$ is used according to \eqref{eq:ref}. Here, $\m{e}_i=\m{x}_i-\m{r}_i$ is the deviation of the system state $\m{x}_i$ from the reference $\m{r}_i$ at time step $i$. $\m{Q}\in\mathbb{R}^{n\times n}$, $\m{Q}=\m{Q}^\intercal\succeq\m{0}$ is a symmetric, positive semidefinite matrix penalizing deviations of the state $\m{x}_i$ from the reference $\m{r}_i$, $\m{R}\in\mathbb{R}^{m\times m}$ and $\m{R}=\m{R}^\intercal\succ\m{0}$ is a symmetric, positive definite matrix penalizing the control effort. Furthermore, {${\gamma \in\left[0,1\right)}$} a discount factor and $c_i \in \mathbb{R}$ denotes the one-step cost.

The choice of our cost function \eqref{eq:Kostena} is motivated as follows: On one hand, the reference to be tracked is usually known only on a finite horizon $N$ (e.g. a road course). On the other hand, the Q-learning method that we use in Section~\ref{section_learning} relies on an infinite horizon cost function which allows the efficient formulation of a Bellman-like equation. Thus, this cost function \eqref{eq:Kostena} accounts for a given reference of finite horizon $N$ and enables stability and convergence of the associated RL algorithm due to the infinite horizon.

Therefore, our problem definition can be summarized in Problem~\ref{problem1}.
\begin{prob}\label{problem1}
	At each time step $k$, find the optimal control input $\m{u}_k^*$ and apply it to system \eqref{eq:System}, where $\m{u}_k^*, \m{u}_{k+1}^*, \dots$ is the control sequence minimizing the discounted cost \eqref{eq:Kostena} subject to the system dynamics \eqref{eq:System}, where $\m{A}$ and $\m{B}$ are unknown, given $\m{x}_k$ and $\m{r}_k, \m{r}_{k+1}, \dots, \m{r}_{k+N}$.
\end{prob}


\section{Extended Q-Function for Reference Tracking}\label{qfunc_section}
Our idea is to define a new reference-dependent Q-function in contrast to the commonly used Q-function (i.e. we define a state-action-\textit{reference} function rather than a state-action function). This Q-function is constructed such that its minimizing control input constitutes a solution to Problem~\ref{problem1}.

In this section, we introduce the reference-dependent Q-function and derive its analytical solution. This provides important insights in how to parametrize the Q-function in Section~\ref{section_learning}, which is the basic ingredient for a convergent reinforcement learning algorithm capable of learning the optimal tracking solution.


We begin with the general case of a \textit{finite} optimization horizon $K$, later we let $K\rightarrow \infty$. The notation seems to be a little bit clumsy but is of high importance, since optimization horizon $K$ and moving horizon of the reference $N$ should not be confused. Furthermore, the notation which control input is plugged into the one-step cost and Q-functions will be useful and we need to distinguish between the actual time step $k$ which the system is in and the time step $\kappa$ on the optimization horizon $K$.

\begin{definition}[Reference-dependent Q-function]\label{def:Q}
	Our proposed reference-dependent Q-function is defined as
	\begin{align}
	\tensor*[^K]{\!Q}{_{\kappa}}&=c_{k+\kappa}+\gamma\sum_{i=k+\kappa+1}^{k+K}\gamma^{i-(k+\kappa+1)}\left.c_i\right|_{\m{u}_i^*} \nonumber \\ &= c_{k+\kappa}+\gamma \left.\tensor*[^K]{\!Q}{_{\kappa+1}}\right|_{\m{u}^*_{k+\kappa+1}}\label{Q_def1}
	\end{align}
	where
	\begin{align}
	c_i&=c\left(\m{x}_{i}, \m{u}_{i}, \m{r}_{i}\right),\label{eq:c}\\
	\tensor*[^K]{\!Q}{_{\kappa}}&=Q_{K-\kappa}\left(\m{x}_{k+\kappa}, \m{u}_{k+\kappa}, \m{r}_{k+\kappa}, \dots, \m{r}_{k+K}\right),\label{eq:Q}\\
	\tensor*[^K]{\!Q}{_{K}}&=c\left(\m{x}_{k+K},\m{u}_{k+K},\m{r}_{k+K}\right)\label{eq:QNN}.
	\end{align}
	Here, $\kappa\in\mathbb{N}_0$, $\kappa < K$ denotes the time step on the current optimization horizon of length $K$ starting at $k$ and $\m{r}_i=\m{0}$, $\forall i>k+N$ (see \eqref{eq:ref}). The notation $\left.c_i\right|_{\m{u}_i^*}$ indicates that the optimal control $\m{u}_i^*$ is applied in \eqref{eq:c} and $\left.\tensor*[^K]{\!Q}{_{\kappa+1}}\right|_{\m{u}^*_{k+\kappa+1}}$ denotes that $\m{u}_{k+\kappa+1}=\m{u}^*_{k+\kappa+1}$ in $\tensor*[^K]{\!Q}{_{\kappa+1}}$ (cf.~\eqref{eq:Q}).
\end{definition}
Therefore, $\tensor*[^K]{\!Q}{_{\kappa}}$ is the accumulated discounted cost from time step $k+\kappa$ to $k+K$ if the control $\m{u}_{k+\kappa}$ is applied at time step $k+\kappa$ and the optimal controls $\m{u}^*_{k+\kappa+1}, \dots, \m{u}^*_{k+K}$ minimizing the cost-to-go are applied thereafter while the reference is provided on a moving horizon of length $N$. With the finite horizon cost function defined as
\begin{align}\label{eq:JkK}
\tensor*[^K]{\!J}{_k} = \sum_{i=k}^{k+K}\gamma^{i-k}\frac{1}{2}\left(\m{e}_i^\intercal\m{Q}\m{e}_i + \m{u}_i^\intercal\m{R}\m{u}_i\right)
\end{align}
the subsequent Lemma~\ref{lemmauQisuJ} follows.
\begin{lemma}\label{lemmauQisuJ}
	The control $\m{u}_k$ minimizing the reference-dependent Q-function $\tensor*[^K]{\!Q}{_{0}}$ is a solution for $\m{u}_k^*$ minimizing $\tensor*[^K]{\!J}{_k}$.
\end{lemma}
\begin{proof}
	With \eqref{eq:JkK} and
	\begin{align}\label{eq:minuQ}
	\min_{\m{u}_k}\tensor*[^K]{\!Q}{_{0}}&=\left.c_k\right|_{\m{u}^*_k}+\gamma \left.\tensor*[^K]{\!Q}{_{1}}\right|_{\m{u}^*_{k+1}}\nonumber \\ &=\left.\tensor*[^K]{\!Q}{_{0}}\right|_{\m{u}^*_k}=\sum_{i=k}^{k+K}\gamma^{i-k}\left.c_i\right|_{\m{u}^*_i}
	\end{align}
	follows
	\begin{align}
	\min_{\m{u}_{k}, \dots, \m{u}_{k+K}}\tensor*[^K]{\!J}{_k}&=\sum_{i=k}^{k+K}\gamma^{i-k}\left.c_i\right|_{\m{u}^*_i}=\left.\tensor*[^K]{\!Q}{_{0}}\right|_{\m{u}^*_k}.
	\end{align}
\end{proof}
Taking the limit
\begin{align}
\lim_{K\rightarrow\infty}\tensor*[^K]{\!J}{_k}=J_k
\end{align}
yields $J_k$ (cf. \eqref{eq:Kostena}). With the Q-function $\tensor*[^K]{\!Q}{_{0}}$ defined according to Definition~\ref{def:Q} and as a result of Lemma~\ref{lemmauQisuJ}, Problem~\ref{problem1} is equivalent to the following Problem~\ref{problem2}.
\begin{prob}\label{problem2}
	Given $\m{r}_k, \m{r}_{k+1}, \dots, \m{r}_{k+N}$, in each state $\m{x}_k$ at time step $k$, find the control $\m{u}_k^*$ minimizing the reference-dependent Q-function $Q_{0}$, where
	\begin{align}
	Q_{0} = \lim_{K\rightarrow\infty}\tensor*[^K]{\!Q}{_{0}}
	\end{align}
	and apply it to the system whose matrices $\m{A}$ and $\m{B}$ are unknown.
\end{prob}

We will proceed in two steps. First, it is assumed that the system matrices $\m{A}$ and $\m{B}$ are known. Later on, this assumption will be dropped and an iterative solution based on a temporal difference error will be introduced. We will now propose the analytical solution of $\tensor*[^K]{\!Q}{_{0}}$ in Theorem~\ref{Theo:Q}.

\begin{theorem}[Analytical solution of {$\tensor*[^K]{\!Q}{_{0}}$}]\label{Theo:Q}
	For $K\geq N$, the Q-function $\tensor*[^K]{\!Q}{_{0}}$ (cf. Definition~\ref{def:Q}) with the objective function \eqref{eq:Kostena} is given by
	\begin{align}
	\tensor*[^K]{\!Q}{_{0}} = \frac{1}{2}\mat{\m{x}_k^\intercal & \m{u}_k^\intercal & \m{r}_k^\intercal & \cdots & \m{r}_{k+N}^\intercal & \m{0}_{\vphantom{k+N}}^\intercal}\m{H}_K\matc{\m{x}_k\\\m{u}_k\\\m{r}_k\\ \vdots \\ \m{r}_{k+N}\\ \m{0}},\label{eq:Q0_mit_H}
	\end{align}
	where $\m{H}_K=\m{H}_K^\intercal\in\mathbb{R}^{\left((K+2)n+m\right)\times\left((K+2)n+m\right)}$ with
	\begingroup
	\setlength{\arraycolsep}{3pt}
	\begin{align}
	\m{H}_K = \matl{\m{h}_{xx} 		& \m{h}_{xu} 		& \m{h}_{xr_0} 		& \m{h}_{xr_1} 		& \m{h}_{xr_2} 		& \cdots & \m{h}_{xr_{K}} 	\\ \label{eq:FormDerQ-Fkt_2}
		\m{h}_{ux} 			& \m{h}_{uu} 		& \m{0} 				& \m{h}_{ur_1}		& \m{h}_{ur_2} 		& \cdots & \m{h}_{ur_{K}} 	\\
		\m{h}_{r_0x} 		& \m{0}					& \m{h}_{r_0r_0}	& \m{0} 				& \m{0} 				& \cdots & \m{0} 				\\
		\m{h}_{r_1x} 		& \m{h}_{r_1u} 		& \m{0}					& \m{h}_{r_1r_1}	& \m{0} 				& \cdots & \m{0} 				\\
		\m{h}_{r_2x} 		& \m{h}_{r_2u} 		& \m{0}					& \m{0}		 			& \m{h}_{r_2r_2}	& \cdots & \m{h}_{r_{2}r_{K}} \\ 
		\m{h}_{r_3x}		& \m{h}_{r_3u}		& \m{0}					& \m{0}					& \m{h}_{r_3r_2}	& \cdots & \m{h}_{r_3r_{K}} 	\\
		\,\vdots 				& \,\vdots 			& \,\vdots 			& \,\vdots 			& \,\vdots 			& \ddots & \,\vdots 			\\ 
		\m{h}_{r_{K}x} 	& \m{h}_{r_{K}u} 	&\m{0}					& \m{0}					& \m{h}_{r_{K}r_2}& \cdots & \m{h}_{r_{K}r_{K}} }.
	\end{align}
	\endgroup
	The exact values of $\m{H}_K$ follow from the subsequent proof.
\end{theorem}
\begin{proof}
	The proof is of rather technical nature and given in \ref{append:notations}.
\end{proof}

For $K\rightarrow \infty$ let $\m{H}$ be the northwestern $\left((N+2)n+m\right)\times\left((N+2)n+m\right)$-submatrix of $\m{H}_K$. Then,
\begin{align}
Q_{0} &= \lim_{K\rightarrow\infty}\tensor*[^K]{\!Q}{_{0}}=\frac{1}{2}\m{z}_k^\intercal\m{H}\m{z}_k\label{eq:H_N_steps}
\end{align}
follows, where $\m{z}_k=\mat{\m{x}_k^\intercal & \m{u}_k^\intercal & \m{r}_k^\intercal & \cdots & \m{r}_{k+N}^\intercal}^\intercal$, as ${\m{r}_i = \m{0}}$, $\forall i>k+N$.

Thus, the Q-function for the LQ tracking problem is quadratic w.r.t. the state $\m{x}_k$, the control input $\m{u}_k$ and the reference $\m{r}_k, \dots, \m{r}_{k+N}$ and furthermore completely parametrized by $\m{H}$. This obviously is a generalization of the well known structure for reference free RL-procedures \cite{Bradtke.1994, Landelius.1997}. In addition, $\m{H}$ is not only quadratic but has a specific structure, which is a new result and allows an efficient parametrization for the Q-learning based algorithm in the following.

\section{Q-Learning Based Tracking}\label{section_learning}
We use the new Q-function for reference tracking on a moving horizon where the system matrices $\m{A}$ and $\m{B}$ are unknown. Thus, our aim is to determine $\m{H}$ (cf. \eqref{eq:H_N_steps}) by observations of states and rewards. 

The optimal control $\m{u}_k^*$, which is equivalent to \eqref{eq:uopt_G} for  $\kappa=0$ and $K\rightarrow\infty$, can be expressed directly by means of $\m{H}$ and is given by Corollary~\ref{corollary:uoptH}. Here, \eqref{dQdu_d2Qdu2} ensures that $\m{u}_k^*$ in fact minimizes the Q-function $Q_{0}$.
\begin{corollary}[Optimal control $\m{u}_k^*$]\label{corollary:uoptH}
	With Lemma~\ref{lemmauQisuJ} and Theorem~\ref{Theo:Q}, the optimal control at time step $k$ is given by
	\begin{align}\label{eq:uoptH}
	\m{u}_{k}^* = -\m{h}_{uu}^{-1}\mat{\m{h}_{ux} & \m{h}_{ur_1}  & \cdots & \m{h}_{ur_{N}}} \matc{\m{x}_{k} \\ \m{r}_{k+1} \\ \vdots \\ \m{r}_{k+N}}.
	\end{align}
\end{corollary}
Hence, if the Q-function $Q_{0}$ is known by means of the matrix $\m{H}$, the optimal control directly results from \eqref{eq:uoptH}.
\begin{note}
	In contrast to usual controllers learned by Reinforcement Learning, our control law \eqref{eq:uoptH} explicitly depends on the reference values $\m{r}_{k+1},\dots,\m{r}_{k+N}$ and is therefore able to generalize to arbitrary reference trajectories on the horizon~$N$.
\end{note}


\subsection{Parametrization of the Reference-dependent Q-function}
In order to learn the Q-function, we parametrize $Q_{0}$ and perform a value iteration on the resulting squared Bellman-like temporal difference (TD) error in order to estimate the Q-function parameters as well as the corresponding optimal control law.
Let the estimated Q-function be parametrized by means of a sum of weighted basis functions:
\begin{align}\label{eq:Qhat}
\hat{Q}_{0}&=\m{w}^\intercal\m{\phi}\left(\m{x}_k, \m{u}_k, \m{r}_k, \dots, \m{r}_{k+N}\right)=\m{w}^\intercal\m{\phi}(\m{z}_k).
\end{align}
Here, $\m{w}\in\mathbb{R}^{{L}}$ is a weight vector and $\m{\phi}:\mathbb{R}^{(N+2)n+m}\rightarrow\mathbb{R}^{{L}}$ is a vector of basis functions. Note that, in contrast to usual Q-function approximations, $\hat{Q}_{0}$ in \eqref{eq:Qhat} explicitly incorporates the reference trajectory $\m{r}_k, \m{r}_{k+1}, \dots, \m{r}_{k+N}$.
\begin{lemma}\label{lemma:Qexakt_long}
	With 
	\begin{align}\label{eq:L}
	L=&~\frac{1}{2}\left((N+2)n+m\right)\left((N+2)n+m+1\right)\nonumber\\&~-\left(n^2(2N-1)+mn\right)
	\end{align}	
	quadratic basis functions ${\m{\phi}=\mat{\phi_1 & \cdots & \phi_L}^\intercal}$, there exists a weight vector $\m{w}=\m{w}^*$ such that $\hat{Q}_{0}=Q_{0}$.
\end{lemma}
\begin{proof}
	Due to the symmetry and the zeros in $\m{H}_K$ (cf. \eqref{eq:FormDerQ-Fkt_2}) and therefore also in $\m{H}$, there are $L$ non-redundant elements in $\m{H}$. Define quadratic basis functions $\phi_l$, $l=1,\dots,L$ of the form
	\begin{align}
	\phi_l= \begin{cases}\left\{\m{z}_k\right\}_i\left\{\m{z}_k\right\}_j, &\text{for } i\neq j \\ \frac{1}{2}\left\{\m{z}_k\right\}_i^2, &\text{for } i=j,\end{cases}
	\end{align}
	where $i, j$ indicate the corresponding non-redundant elements of $\m{H}$, $\left\{\m{\cdot}\right\}_i$ denotes the $i$-th element of a vector and $\m{z}_k$ is defined as in \eqref{eq:H_N_steps}. 
	Thus, $Q_{0}$ in \eqref{eq:H_N_steps} is equivalent to $\hat{Q}_{0}$ in \eqref{eq:Qhat} if the weights $\left\{\m{w}\right\}_i$, $i=1, \dots, L$, are equal to the corresponding non-redundant elements in $\m{H}$ which we denote by $\m{w}=\m{w}^*$.
\end{proof}

Although $L$ in Lemma~\ref{lemma:Qexakt_long} gives the maximum number of weights needed in order to parametrize $Q_{0}$ exactly, in the frequently occurring case of a sparse weighting matrix $\m{Q}$ in the cost functional \eqref{eq:Kostena}, the exact knowledge of the structure of $\m{H}$ can be exploited further in order to drastically reduce the weights $\m{w}$ that are necessary.
\begin{lemma}\label{lemma:Qexakt}
	If the $l$-th row and $l$-th column of $\m{Q}$ equals zero, then the
	\begin{itemize}
		\item $l$-th column of $\m{h}_{xr_i} \forall i \in \left\{0,\dots,N\right\}$,
		\item $l$-th column of $\m{h}_{ur_i} \forall i \in \left\{0,\dots,N\right\}$ and
		\item $l$-th row and $l$-th column of $\m{h}_{r_ir_j} \forall i,j \in \left\{0,\dots,N\right\}$
	\end{itemize}
	are all equal to zero. Thus, the number of non-redundant weights $L$ in $\m{H}$ (corresponding to the weight vector $\m{w}$) reduces to
	\begin{align}\label{eq:L_sparse}
	L &= (n-q)(n-q+1)\left(\frac{N}{2}+1\right)+\frac{1}{2}n(n+1)\\ \nonumber&+\left(m+N(n-q)\right)(m+n)+\frac{1}{2}(N-2)(N-1)(n-q)^2,
	\end{align}
	where $\m{h}_{xr_0}=\m{Q}$ and $\m{h}_{r_0r_0}=-\m{Q}$ has been considered and $q$ denotes the number of rows and columns of $\m{Q}$ that are both zero.
\end{lemma}
\begin{proof}
	The rather technical proof directly follows from \eqref{eq:Q_kappa} considering \eqref{eq_rho0}--\eqref{eq_xi}.
\end{proof}
Although Lemma~\ref{lemma:Qexakt} is of technical nature, sparsity of $\m{H}$ is essential in order to implement efficient learning controllers as will be discussed in Section~\ref{sectionsimulation}.
\begin{note}
	Based on the specific structure of the analytical solution of $Q_{0}$ derived in Theorem~\ref{Theo:Q}, according to Lemma~\ref{lemma:Qexakt}, $L$ quadratic basis functions are sufficient to parametrize $Q_{0}$ exactly.
\end{note}
\begin{note}
	Equation \eqref{eq:Qhat} and Lemma~\ref{lemma:Qexakt} show that our proposed Q-function generalizes over reference trajectories as the weight vector $\m{w}$ does not depend on a specific reference we intend to follow.
\end{note}

\subsection{Online Learning Algorithm}
In this section, we propose our value iteration based algorithm in order to learn the weights $\m{w}$ online by minimizing the squared temporal difference error.
Let 
\begin{align}\label{eq:Qhat1}
\hat{Q}_{1}&=\m{w}^\intercal\m{\phi}\left(\m{x}_{k+1}, \m{u}_{k+1}, \m{r}_{k+1}, \dots, \m{r}_{k+N}, \m{0}\right)\nonumber \\&=\m{w}^\intercal\m{\phi}\left(\m{z}_{k+1}\right),
\end{align}
where $\m{z}_{k+1}=\mat{\m{x}_{k+1}^\intercal & \m{u}_{k+1}^\intercal & \m{r}_{k+1}^\intercal & \cdots & \m{r}_{k+N}^\intercal & \m{0}_{\vphantom{k+N}}^\intercal}^\intercal$.

The estimated optimal control input $\hat{\m{u}}_{k+\kappa}^*$ based on the estimated Q-function $\left.\hat{Q}_{\kappa}\right|_{\m{u}_{k+\kappa}}$ is defined as
\begin{align}\label{eq:opt_u_estimated}
\hat{\m{u}}_{k+\kappa}^*&=\argmin_{\m{u}_{k+\kappa}}\left.\hat{Q}_{\kappa}\right|_{\m{u}_{k+\kappa}}\nonumber \\ &=\m{L}\mat{\m{x}_{k+\kappa}^\intercal&\m{r}_{k+\kappa+1}^\intercal & \dots & \m{r}_{k+\kappa+N}^\intercal}^\intercal,
\end{align}
where $\m{L}=-\hat{\m{h}}_{uu}^{-1}\mat{\hat{\m{h}}_{ux} & \hat{\m{h}}_{ur_1}  & \cdots & \hat{\m{h}}_{ur_{N}}}$ (cf. \eqref{eq:uoptH}). Here, let $\hat{\m{H}}$ be the estimated matrix $\m{H}$ based on $\m{w}$ and $\hat{\m{h}}$ denote the submatrices of $\hat{\m{H}}$ as in \eqref{eq:FormDerQ-Fkt_2}. Furthermore, $\hat{\m{u}}_{k+\kappa}^*=\m{u}_{k+\kappa}^*$ results if $\m{w}=\m{w}^*$ according to Lemma~\ref{lemma:Qexakt} and the optimal control law $\m{L}^*$ results if $\hat{\m{H}}=\m{H}$.


Then, based on the Bellman-like equation \eqref{Q_def1}, the TD error $\epsilon_k$ \cite{Sutton.1988} is given in Definition~\ref{def:TD}. If ${\m{w}=\m{w}^*}$, $\epsilon_k$ would vanish.

\begin{definition}[TD error]\label{def:TD}
	The temporal difference error $\epsilon_k$, i.e. the approximation error of the Bellman-like equation \eqref{Q_def1} due to the deviation of the weight estimate $\m{w}$ from $\m{w}^*$ is defined as
	\begin{align}\label{eq:td}
	\epsilon_k&=c_k+\gamma\left.\hat{Q}_{1}\right|_{\hat{\m{u}}_{k+1}^*}-\hat{Q}_{0}\nonumber \\ &=c_k+\gamma\m{w}^\intercal \phi\left(\m{z}_{k+1}^*\right)-\m{w}^\intercal \phi\left(\m{z}_{k}\right).
	\end{align}
	Here, $\m{z}_{k+1}^* = \left.\m{z}_{k+1}\right|_{\m{u}_{k+1}=\hat{\m{u}}_{k+1}^*}$.	
\end{definition}



In order to improve the estimated reference-dependent Q-function $\hat{Q}_{0}$ as well as the resulting estimated optimal control $\hat{\m{u}}_k^*$, we employ a value iteration procedure. This iteration consists of a \textit{policy evaluation} which updates the weight estimate $\m{w}^{(i)}$ representing the Q-function and a \textit{policy improvement} step, where, based on the updated Q-function weight $\m{w}^{(i)}$ and the corresponding matrix $\m{H}^{(i)}$, the control law $\m{L}^{(i)}$ is adapted according to \eqref{eq:opt_u_estimated}.

To evaluate $\epsilon_k$ in \eqref{eq:td}, $\hat{\m{u}}_{k+1}^*$ is required. However, as the optimal weight $\m{w}^*$ is unknown a priori, we initialize $\m{w}^{(0)}=\m{0}$ and the estimated optimal control by $\hat{\m{u}}_{k+1}^{*(0)} = \m{0}$ which is achieved by setting $\m{L}^{(0)}=\m{0}$.

In the \textit{policy evaluation} step, the aim is to find an updated $\m{w}^{(i+1)}$ such that $\left({\epsilon}_k^{(i)}\right)^2$ is minimized, where\tdd{zu value iteration ändern!!!}
\begin{align}\label{eq:td_VI}
{\epsilon}_k^{(i)} &= c_k + \gamma{\m{w}^{(i)}}^\intercal\m{\phi}\left(\m{z}_{k+1}^{*{(i)}}\right)-{\m{w}^{(i+1)}}^\intercal\m{\phi}\left(\m{z}_{k}\right)
\end{align}
in analogy to \eqref{eq:td}. In accordance with Lemma~\ref{lemma:Qexakt}, $\m{w}\in\mathbb{R}^L$ follows. Thus, ${\epsilon}_k^{(i)}$ needs to be considered at $M\geq L$ time steps in order to perform a least-squares update, where $M$ is the number of samples used for the policy evaluation. Then, $\m{w}^{(i+1)}$ results from
\begin{align}\label{eq:PE}
\m{w}^{(i+1)}=\argmin_{\m{w}^{(i+1)}}\sum_{j=k-M+1}^{k}\left({\epsilon}_j^{(i)}\right)^2.
\end{align}
Now, we define
\begin{align}\label{eq:Phimatrix}
\m{\Phi}&=\matc{\m{\phi}\left(\m{z}_{k-M+1}\right)& \dots & \m{\phi}\left(\m{z}_{k}\right)}^\intercal
\end{align}
and
\begin{align}\label{eq:cmatrix}
\m{c}=\matc{c_{k-M+1}+\gamma{\m{w}^{(i)}}^\intercal\phi\left(\m{z}_{k-M+2}^{*(i)}\right)\\ \vdots \\ c_k+\gamma{\m{w}^{(i)}}^\intercal\phi\left(\m{z}_{k+1}^{*(i)}\right)}.
\end{align}
If the excitation condition
\begin{align}\label{rankcondition}
\rank\left({\m{\Phi}^\intercal}\m{\Phi}\right) = L
\end{align}
is satisfied, $\m{w}^{(i+1)}$ minimizing \eqref{eq:PE} exists, is unique and given by
\begin{align}\label{eq:LS}
\m{w}^{(i+1)}&=\left({\m{\Phi}^\intercal}\m{\Phi}\right)^{-1}{\m{\Phi}^\intercal}\m{c}
\end{align}
(cf. \cite[Theorem~2.1]{Astrom.1995}).

Then, the \textit{policy improvement} step is based on the new weight estimate $\m{w}^{(i+1)}$ and its corresponding $\m{H}^{(i+1)}$ and results in 
\begin{align}\label{eq:policy_improvement}
\m{L}^{(i+1)}=-{\m{h}_{uu}^{(i+1)}}^{-1}\mat{\m{h}_{ux}^{(i+1)} & \m{h}_{ur_1}^{(i+1)}  & \cdots & \m{h}_{ur_{N}^{(i+1)}}}
\end{align}
(cf. Corollary~\ref{corollary:uoptH}).

With the time step still being fixed at $k$, this iteration is performed until the change in $\m{w}$ stays below a given threshold $e_{\m{w}}$\tdd{ or a maximum number of iterations $i_\text{max}$ is reached}.
Although we would like to evaluate the Q-function corresponding to the target policy $\hat{\m{u}}_k^*$ (note that $\m{z}_{k+1}^{*(i)}$ is used in \eqref{eq:td_VI}), we have not yet discussed how to choose the behavior policy, i.e. the control $\m{u}_k$ that is applied to the system and appears in $c_k$ and $\m{z}_k$ (cf. \eqref{eq:td_VI}). 
During the learning process is active, let
\begin{align}\label{utilde}
\m{u}_k=\tilde{\m{u}}_k^*=\hat{\m{u}}_k^*+\m{\xi}.
\end{align}
Here, the Gaussian noise $\m{\xi} \sim \mathcal{N}_{m}\left(\m{0},\m{\Sigma}\right)$ serves as exploration noise as persistent excitation is required for convergence (cf. \eqref{rankcondition}, \cite{Narendra.1987, Astrom.1995, Lewis.2009}).
Furthermore, if the reference to track is smooth,
additional excitation noise should be added to the reference in order to satisfy condition \eqref{rankcondition}. The complete algorithm is shown in Algorithm~\ref{algorithm} and learns to track an arbitrary reference that is given on a moving horizon of length $N$ without knowledge of the system matrices $\m{A}$ and $\m{B}$. If the system dynamics is not expected to change over time, learning might be stopped after the first complete value iteration based on $M$ data tuples has been performed. This can be done by stopping Algorithm~\ref{algorithm} after line 11. In practice, learning might be enabled again whenever the Bellman error increases which is an indicator for suboptimal weights $\m{w}$.

\begin{algorithm}[tb!]
	\caption{Q-function Tracking Controller}\label{algorithm}
	\begin{algorithmic}[1]
		\Initialize{$M$, $\m{w}=\m{w}^{(0)}=\m{0}$, $\m{L}^{(0)}=\m{0}$}
		\For{$k=1,2,\dots$}
		\State apply $\tilde{\m{u}}_k^*$ \eqref{utilde} to the system
		\If{$k\!\mod M=0$}
		\State $i=0, \m{w}^{(i)}=\m{w}$
		\Do
		\State policy evaluation: $\m{w}^{(i+1)}$ \eqref{eq:PE}
		\State policy improvement: $\m{L}^{(i+1)}$ \eqref{eq:policy_improvement}
		\State $i = i+1$
		\doWhile{$\norm{\m{w}^{(i)}-\m{w}^{(i-1)}}_2>e_{\m{w}}$}
		\State $\m{w}=\m{w}^{(i)}$
		\EndIf
		\EndFor
	\end{algorithmic}
\end{algorithm}
\begin{note}
	The iterative procedure of policy evaluation and policy improvement in Algorithm~\ref{algorithm} is a value iteration (cf.~\cite{Sutton.2018}). This is due to the definition of ${\epsilon}_k$ in \eqref{eq:td_VI}, where $\hat{Q}_{1}$ relies on $\m{w}^{(i)}$ and $\hat{Q}_{0}$ on $\m{w}^{(i+1)}$ (cf.~\cite{Lewis.2009}).	
\end{note}
\begin{note}
	Just as regular Q-learning \cite{Watkins.1992}, our algorithm belongs to the off-policy RL methods \cite{Sutton.2018} as the behavior policy $\tilde{\m{u}}_k^*=\hat{\m{u}}_k^*+\m{\xi}$ is followed while the agent learns the Q-function belonging to the target policy $\hat{\m{u}}_k^*$.
\end{note}


\subsection{Convergence Analysis of the Learning Algorithm}\label{convergence}
In this section, we will provide convergence proofs, i.e. show that the estimated reference-dependent Q-function $\hat{Q}_{0}$ converges to the underlying Q-function $Q_{0}$ and that ${\m{w}^{(i)}\rightarrow \m{w}^*}$, i.e. $\m{H}^{(i)}\rightarrow \m{H}$ as $i\rightarrow \infty$. This also implies that $\m{L}^{(i)}\rightarrow \m{L}^*$, i.e. the value iteration converges to the optimal control law.

Our convergence analysis is structured as follows. First, we prove that the value iteration (i.e. iterating between \eqref{eq:PE} and \eqref{eq:policy_improvement}) is equivalent to a matrix sequence on $\m{H}^{(i)}$. Second, we prove that this sequence of matrices is upper bounded in the sense that $\m{0}\preceq\m{H}^{(i)}\preceq \m{Y}$ while $\m{H}^{(i)}$ is monotonically increasing, i.e. $\m{0}\preceq\m{H}^{(i)}\preceq\m{H}^{(i+1)}$. Hence, the sequence converges. Finally, we show that the converged sequence fulfills the Bellman equation and the corresponding control law is optimal.

The following Lemma~\ref{lemma:VI_matrix_iteration} extends \cite[Lemma~1]{AlTamimi.2007} to the tracking case and shows that our proposed value iteration is equivalent to a matrix sequence on $\m{H}^{(i)}$.

\begin{lemma}\label{lemma:VI_matrix_iteration}
	Let $\m{H}^{(0)} = \m{0}$, $\m{R}^\intercal = \m{R} \succ \m{0}$, $\m{Q}^\intercal = \m{Q} \succeq \m{0}$ and $(\m{A}, \m{B})$ controllable. The value iteration described by \eqref{eq:PE} and \eqref{eq:policy_improvement} is equivalent to the iteration
	\begin{align}\label{eq:iteration_Hi1}
	\m{H}^{(i+1)}=\m{G}+\gamma\m{M}\!\left(\m{L}^{(i)}\right)^\intercal \m{H}^{(i)}\m{M}\!\left(\m{L}^{(i)}\right),
	\end{align}
	where 
	\begin{align}
	\m{G}=\mat{\m{Q}&\m{0}&-\m{Q}&\m{0}\\ \m{0}&\m{R}&\m{0}&\m{0}\\ -\m{Q}&\m{0}&\m{Q}&\m{0}\\ \m{0}&\m{0}&\m{0}&\m{0}}
	\end{align}
	and
	\begin{align}\label{eq:M}
	\m{M}\!\left(\m{L}^{(i)}\right)&=\matc{\m{A}&\m{B}&\m{0}&\m{0}&\m{0}&\cdots&\m{0}\\
		\m{L}_x^{(i)}\m{A}&\m{L}_x^{(i)}\m{B}&\m{0}&\m{0}&\m{L}_1^{(i)}&\cdots&\m{L}_{N-1}^{(i)}\\
		\m{0}&\m{0}&\m{0}&\m{I}_n&\m{0}&\cdots&\m{0}\\
		\m{0}&\m{0}&\m{0}&\m{0}&\m{I}_n&\cdots&\m{0}\\
		\vdots&\vdots&\vdots&\vdots&\vdots&\ddots&\vdots\\
		\m{0}&\m{0}&\m{0}&\m{0}&\m{0}&\cdots&\m{I}_n\\
		\m{0}&\m{0}&\m{0}&\m{0}&\m{0}&\cdots&\m{0}}
	\end{align}
	with $\m{L}^{(i)}=\mat{\m{L}^{(i)}_x&\m{L}^{(i)}_1&\cdots&\m{L}^{(i)}_N}=\m{L}\!\left(\m{H}^{(i)}\right)$ depending on $\m{H}^{(i)}$ as in \eqref{eq:policy_improvement}.
\end{lemma}
\begin{proof}
	See \ref{append:VI_matrix_iteration}.
\end{proof}
The following technical Lemma~\ref{lemma:hilfslemma} is required later for the proof of Lemma~\ref{lemma:monotonically_increasing}.
\begin{lemma}\label{lemma:hilfslemma}
	For $\m{H}^{(0)} = \m{0}$, $\m{R}^\intercal = \m{R} \succ \m{0}$ and $\m{Q}^\intercal = \m{Q} \succeq \m{0}$, $\forall i > 0$: $\m{L}^{(i)}\mat{\m{x}_{k+\kappa}^\intercal&\m{r}_{k+\kappa+1}^\intercal & \dots & \m{r}_{k+\kappa+N}^\intercal}$ is the unique minimizer of
	\begin{align}
	\hat{Q}_{\kappa}^{(i)}=\frac{1}{2}\m{z}_{k+\kappa}^\intercal\m{H}^{(i)}\m{z}_{k+\kappa}.
	\end{align}
\end{lemma}
\begin{proof}
	Due to $\m{H}^{(0)}=\m{0}\succeq\m{0}$, $\m{R} \succ \m{0}$ and $\m{Q} \succeq \m{0}$, \eqref{eq:iteration_Hi1} yields $\m{H}^{(i)}\succeq\m{0}$. Furthermore, due to $\m{R} \succ \m{0}$ it is obvious that $\m{h}_{uu}^{(i)}\succ\m{0}$, $\forall i>0$. Therefore, $\frac{\partial \hat{Q}_{\kappa}^{(i)}}{\partial \m{u}_{k+\kappa}}=\m{0}$ yields $\m{L}^{(i)}$ and $\frac{\partial ^2 \hat{Q}_{\kappa}^{(i)}}{\partial \m{u}_{k+\kappa}^2}\succ\m{0}$ follows due to $\m{h}_{uu}^{(i)}\succ\m{0}$ and completes the proof.
\end{proof}

Define the operator
\begin{align}\label{eq:operator}
F\left(\m{\Omega}^{(i)}, \m{\Gamma}^{(i)}\right)&=\m{G}+\gamma\m{M}\!\left(\m{\Gamma}^{(i)}\right)^\intercal \m{\Omega}^{(i)}\m{M}\!\left(\m{\Gamma}^{(i)}\right),
\end{align}
i.e. $F\left(\m{H}^{(i)}, \m{L}^{(i)}\right)=\m{H}^{(i+1)}$ according to \eqref{eq:iteration_Hi1}.

In order to prove that $\m{H}^{(i)}$ given according to Lemma~\ref{lemma:VI_matrix_iteration} is upper bounded, the following technical Lemma~\ref{lemma:monotonically_increasing} is required first, which generalizes \cite[Lemma~B.1.1]{Landelius.1997} to cope with the reference-dependent Q-function. Note that knowledge of the exact structure of $\m{H}$ and therefore the analytical solution by means of Theorem~\ref{Theo:Q} plays a crucial role for the extension to the tracking case.

\begin{lemma}\label{lemma:monotonically_increasing}
	Let $\m{W}^{(i)}$ be an arbitrary matrix sequence and $\m{0}\preceq\m{H}^{(0)}\preceq\m{Z}^{(0)}$. Then, given the sequences $\m{Z}^{(i+1)}=F\left(\m{Z}^{(i)},\m{W}^{(i)}\right)$ and $\m{H}^{(i+1)}=F\left(\m{H}^{(i)},\m{L}\left(\m{H}^{(i)}\right)\right)$ with \eqref{eq:operator} it follows that $\m{0}\preceq\m{H}^{(i+1)}\preceq\m{Z}^{(i+1)}$.
\end{lemma}
\begin{proof}
	See \ref{append:monotonically_increasing}.
\end{proof}
In the next step, upper boundedness of $\m{H}^{(i)}$ in the sense that $\m{0}\preceq\m{H}^{(i)}\preceq \m{Y}$ is shown. For the regulation case, boundedness was shown in \cite[Lemma~B.1.2]{Landelius.1997}. In contrast to that, we consider the specific structure of the iteration in the tracking case (cf. \eqref{eq:iteration_Hi1}) and prove boundedness of $\m{H}^{(i)}$ for the more generalized tracking formulation. For reasons of self-consistency, the complete proof is given.

\begin{lemma}\label{lemma:upper_bounded}
	Let $\left(\m{A}, \m{B}\right)$ be controllable and $\m{H}^{(i)}$ be the sequence \eqref{eq:iteration_Hi1} with $\m{H}^{(0)}=\m{0}$. Then, there exists $\m{Y}$ such that $\m{0}\preceq\m{H}^{(i)}\preceq\m{Y}$.
\end{lemma}
\begin{proof}
	See \ref{append:upper_bounded}.
\end{proof}
We will now get to the main result of our convergence analysis and show that the proposed value iteration converges to the optimal weight vector $\m{w}^*$ and the optimal control law $\m{L}^*$ in the tracking case described by Problem~\ref{problem1}.
\begin{theorem}
	Let $\m{R}^\intercal = \m{R} \succ \m{0}$, $\m{Q}^\intercal = \m{Q} \succeq \m{0}$, $(\m{A}, \m{B})$ controllable and $\m{w}^{(0)}=\m{0}$, i.e. $\m{H}^{(0)} = \m{0}$. Then, iterating between \eqref{eq:PE} and \eqref{eq:policy_improvement} yields $\m{H}^{(i)}\rightarrow\m{H}$, i.e. $\m{w}^{(i)}\rightarrow\m{w}^*$ as well as $\m{L}^{(i)}\rightarrow\m{L}^*$.
\end{theorem}
\begin{proof}
	According to Lemma~\ref{lemma:VI_matrix_iteration}, the value iteration is equivalent to iterating on $\m{H}^{(i)}$ (cf. \eqref{eq:iteration_Hi1}). With $\m{Z}^{(0)}=\m{H}^{(0)}$ and $\m{Z}^{(i+1)}=F\left(\m{Z}^{(i)}, \m{L}\left(\m{H}^{(i+1)}\right)\right)$, Lemma~\ref{lemma:monotonically_increasing} yields $\m{0}\preceq\m{H}^{(i)}\preceq\m{Z}^{(i)}$. With $\m{H}^{(0)}=\m{0}$ follows $\m{H}^{(1)}=\m{G}\succeq\m{0}$ and hence $\m{H}^{(1)}-\m{Z}^{(0)}\succeq\m{0}$. 
	The proof is drawn by induction: Assume the induction hypothesis $\m{H}^{(i)}-\m{Z}^{(i-1)}\succeq\m{0}$. Then,
	\begin{align}
	&\m{H}^{(i+1)}-\m{Z}^{(i)}\nonumber \\ &=\gamma \m{M}\left(\m{L}\left(\m{H}^{(i)}\right)\right)^\intercal\left(\m{H}^{(i)}-\m{Z}^{(i-1)}\right)\m{M}\left(\m{L}\left(\m{H}^{(i)}\right)\right)\nonumber \\&\succeq\m{0}.
	\end{align}
	This implies $\m{0}\preceq\m{H}^{(i)}\preceq\m{Z}^{(i)}\preceq\m{H}^{(i+1)}$.
	
	As the matrix sequence is upper bounded by $\m{Y}$ according to Lemma~\ref{lemma:upper_bounded} and $\m{0}\preceq\m{H}^{(i)}\preceq\m{H}^{(i+1)}$, the limit $\m{H}^{(\infty)}$ exists, i.e. the value iteration converges to
	\begin{align}
	\m{H}^{(\infty)}=\m{G}+\gamma\m{M}\left(\m{L}\left(\m{H}^{(\infty)}\right)\right)^\intercal\m{H}^{(\infty)}\m{M}\left(\m{L}\left(\m{H}^{(\infty)}\right)\right).
	\end{align}
	Furthermore, $\m{L}\left(\m{H}^{(\infty)}\right)$ minimizes $\hat{Q}_{0}^{(\infty)}$ according to Lemma~\ref{lemma:hilfslemma}. Thus, 
	\begin{align}
	\lim_{i\rightarrow \infty}\epsilon_k^{(i)}&=c_k+\gamma{\m{z}_{k+1}^{*}}^\intercal\m{H}^{(\infty)}\m{z}_{k+1}^{*}\nonumber \\ &-\m{z}_k^\intercal \, F\left(\m{H}^{(\infty)}, \m{L}\left(\m{H}^{(\infty)}\right)\right)\m{z}_k = 0.
	\end{align}
	Therefore, $\m{H}^{(i)}\rightarrow\m{H}^{(\infty)}=\m{H}$, i.e. $\m{w}^{(i)}\rightarrow\m{w}^*$ and $\m{L}^{(i)}\rightarrow\m{L}^*$ which completes the proof.
	
\end{proof}

\section{Results}\label{simulationResults}\label{sectionsimulation}

In order to validate our proposed method, we show results for two LQ-tracking problems with unknown system dynamics. We furthermore compare our results with an RL tracking method which assumes that the reference can be described by a time-invariant exo-system $\m{f}_{\text{ref}}$.

For excitation purposes (cf. \eqref{rankcondition}), the variance of the input excitation $\m{\xi}$ in \eqref{utilde} is set to $\m{\Sigma}=0.1$ and Gaussian noise $\xi_{\text{ref}}\sim \mathcal{N}\left({0},0.1\right)$ is added to the reference trajectory which we intend to track while learning is active. We furthermore choose the number of data tuples to $M=1.2L$, with $L$ according to \eqref{eq:L_sparse}, the stopping criterion to $e_{\m{w}}=\SI{1e-6}{}$, the discount factor to $\gamma = 0.9$, the horizon on which the reference is known to $N=10$ and finally $\m{x}_0= \m{0}$ for all experiments.

\subsection{Simulation Examples}
The first system is a second-order rotatory mass-spring-damper system, whereas the second system is a sixth-order linear single-track steering model. Both systems were initially given in continuous time and hence discretized by means of a Tustin approximation with a sampling time of $\SI{1}{\second}$.
\subsubsection{System~1}
Consider a rotatory mass-spring-damper system modeled by the discrete-time second order linear state space representation
\begin{align}\label{sys_diskret}
\m{x}_{k+1}=\matc{0.99&0.9\\-0.02&0.8}\m{x}_k+\matc{0.01\\0.02}u_k.
\end{align}

The control input $u_k$ is a torque command applied to the system. Note that this system is not known to the controller and only needed for simulation as well as validation purposes. Furthermore, let
\begin{align}
\m{Q}=\text{diag}(100, 0) \text{ and } \m{R}=1
\end{align}	
be the parameters of the cost function \eqref{eq:Kostena}, i.e. we focus on tracking a reference angle $x_{1,\text{ref}}=\alpha_\text{ref}$ that is given on a moving horizon of length $N=10$.

\subsubsection{System~2}
Our second example system is a sixth order linear single-track steering model
\begin{equation}\label{continuous_ss}
\m{x}_{k+1} = {\m{A}}\m{x}_k + {\m{B}}u_k
\end{equation}
with
\begin{equation*}\setlength{\arraycolsep}{3pt}
{\m{A}} = \begin{bmatrix*}[r]
-0.741  		&-0.033	&0		&0 	&-\SI{2.0E-4}{}	&-\SI{6.4E-3}{}\\
4.146	&-0.914	&0 		&0	&\SI{4.7E-3}{}	&0.151\\
2.073								&0.043				&1		&0	&\SI{2.4E-3}{}	&0.076\\
23.326							&0.106											&20	&1	&0.022	&0.693\\
23.499		&-2.593		&0		&0	&-0.939 &-2.053\\
11.749								&-1.297											&0		&0	&0.031	&-0.027
\end{bmatrix*}
\end{equation*}
and
\begin{align*}\setlength{\arraycolsep}{3pt}
{\m{B}} = \begin{bmatrix}
-\SI{3.8E-4}{} & 0.0091 & 0.0046 & 0.041 & 0.117  & 0.059
\end{bmatrix}^\intercal,
\end{align*}
where $u_k$ is a torque applied to the steering wheel by the controller. The parameters of the corresponding continuous-time model can be found in \cite{Kopf.2018}.
The system state $\m{x}_k$ is given by
\begin{align} \label{eq:Zustaende Fzg}
\m{x}_k = \begin{bmatrix}
\beta_k & {\psi}_{r,k} & \psi_k & y_k & {\delta}_{v,k} & \delta_k
\end{bmatrix}^\intercal,
\end{align}
with sideslip angle $\beta_k$, yaw angle $\psi_k$, yaw rate ${\psi}_{r,k}$, lateral deviation from the origin $y_k$, steering wheel angle $\delta_k$ and angular velocity ${\delta}_{v,k}$. Geometric relations are depicted in Fig.~\ref{fig:einspurmodell}, where $v$ is the constant velocity ($\SI{20}{\meter\per \second}$ in our example) and $\delta_s = 0.0625\delta$ denotes the steering angle in contrast to the steering wheel angle $\delta$.
\begin{figure}[b!]
\centering
\footnotesize
\usetikzlibrary{fit}
\usetikzlibrary{calc}
\usetikzlibrary{positioning}
\usetikzlibrary{shapes,snakes}
\usetikzlibrary{backgrounds}
\usetikzlibrary{arrows}

\begin{tikzpicture}[scale=0.6]
    \draw [fill=black!50, rounded corners, rotate around={0:(0,0)}] (-1,-0.5) rectangle (1,0.5);
    \draw [fill=black!50, rounded corners, rotate around={-33.69:(6.5,0)}] (5.5,-0.5) rectangle (7.5,0.5);
    \draw [fill=black!80, rotate around={0:(0,0)}] (0,-0.05) rectangle (6.5,0.05);
    \draw [fill=black, rotate around={0:(0,0)}] (3,0) ellipse (.1 and 0.1);

    
%
    \draw [dotted, rotate around={0:(0,0)}] (6.5,0) -- (8.2,0);
    \draw [dotted] (6.5,0) -- (7.91449,-0.94299);
    \draw [] (8.2,0) arc (0:-33.69:1.7);
    \draw [] (7.85,-0.35) node {$\delta_s$};
    \draw [-latex, rotate around={0:(0,0)}] (3,0) -- (5.6,-0.64);
    \draw [] (5.4,-0.9) node {$v$};
    \draw [rotate around={0:(0,0)}] (5.6,0) arc (0:-14:2.6);
	\draw [] (5.25,-0.3) node[] {$\beta$};
	\draw [->, rotate around={0:(0,0)}] (3,-0.5) arc (270:315:0.5);
	\draw [rotate around={0:(0,0)}] (3,-0.5) arc (270:135:0.5);  
	\draw [] (2.4,-0.7) node {${{\psi}_r}$};   
	
	\draw [-latex] (-2,-1.5) -- (-2,0);
	\draw [] (-1.6,-0.75) node {$y$};  
	\draw [] (-2.1,-1.5) -- (-1.9,-1.5);
	\draw [] (-2.4,-1.5) node {$0$};  
\end{tikzpicture}
\caption{Geometric relations of the single-track model.} \label{fig:einspurmodell}
\end{figure}
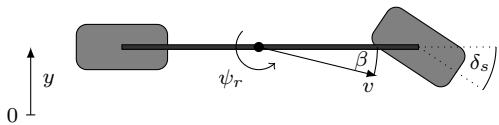
As we desire to track the lateral position $x_{4,\text{ref}} = y_{\text{ref}}$, we choose
\begin{align}
\m{Q}=\text{diag}(0, 0, 0, 100, 0, 0) \text{ and } \m{R}=1.
\end{align}	

\subsection{Evaluation Method}
In order to compare our method with the class of RL tracking algorithms where the reference trajectory is assumed to be generated by a time-invariant exo-system $\m{f}_{\text{ref}}(\m{r}_k)=\m{F}_{\text{ref}}\m{r}_k$, let both the reference angle $\alpha_{\text{ref}}$ (system~1) and the reference lateral position $y_{\text{ref}}$ (system~2) be generated by
\begin{align}\label{eq:refe}
\m{r}_{k+1} &= \underbrace{\mat{0.9801&0.1987\\-0.1987&0.9801}}_{\m{F}_{\text{ref}}}\m{r}_k, \\
\alpha_{\text{ref}} &= y_{\text{ref}} = \mat{1&0}\m{r}_k\label{eq:refe1}
\end{align}
during the training procedure. For comparison reasons, we then train both our proposed method and an algorithm as in \cite{Luo.2016,Kiumarsi.2014} which assumes that the reference always follows the dynamics of an exo-system $\m{f}_{\text{ref}}(\m{r}_k)$ (in the following termed \textit{comparison algorithm}). The comparison algorithm is trained  on the augmented system
\begin{align}\label{eq:augmented_sys}
\m{x}_{\text{aug},k+1}=\matc{\m{A}&\m{0}\\\m{0}&\m{F}_{\text{ref}}}\m{x}_{\text{aug},k}+\matc{\m{B}\\ \m{0}}\m{u}_k,
\end{align}
where $\m{x}_{\text{aug},k}=\mat{\m{x}_k^\intercal&\m{r}_k^\intercal}^\intercal$.
After the training, we vary $\alpha_{\text{ref}} = y_{\text{ref}}$ (arbitrary references such as different frequencies, ramps and steps) in order to show that the controller learned with our method successfully generalizes to these references. In contrast, we show the resulting behavior of the comparison algorithm which is constructed on the assumption that the reference dynamics always follows \eqref{eq:refe}. 

Our evaluation is twofold. On one hand, after the learning process has finished, we analyze the root mean square (RMS) tracking errors $\alpha_{\text{RMS}}$ and $y_{\text{RMS}}$ between the learned tracking behavior by means of the trajectory $\alpha_{\text{learned}}$ and $y_{\text{learned}}$ where the system dynamics is unknown (both for our algorithm and the comparison algorithm) and the optimal solution $\alpha_{\text{opt}}$ and $y_{\text{opt}}$ which results from Theorem~\ref{Theo:Q} and known system dynamics.
On the other hand, we compare the learned weights $\m{w}$ of our algorithm with the optimal solution $\m{w}^*$ (weights corresponding to Theorem~\ref{Theo:Q} and Lemma~\ref{lemma:Qexakt}). In order to achieve comparability for different ranges of $\m{w}$, we normalize the absolute error of each weight with the maximum absolute weight $\max_j \abs{\left\{\m{w}^*\right\}_j}$ and define the average of this normalized absolute error by
\begin{align}\label{eq:e1}
e_{\text{I}}&=\frac{1}{L}\sum_{i=1}^{L}\frac{\abs{\left\{\m{w}\right\}_i-\left\{\m{w}^*\right\}_i}}{\max_j \abs{\left\{\m{w}^*\right\}_j}}.
\end{align}
Its maximum is given by
\begin{align}\label{eq:e2}
e_{\text{II}}&=\max_i \frac{\abs{\left\{\m{w}\right\}_i-\left\{\m{w}^*\right\}_i}}{\max_j \abs{\left\{\m{w}^*\right\}_j}}.
\end{align}
Note that these measures are only reasonable to judge how well our method has learned the unknown optimal weights $\m{w}^*$ as the comparison algorithm does not learn $\m{w}^*$ corresponding to Theorem~\ref{Theo:Q} and Lemma~\ref{lemma:Qexakt} but the optimal weights corresponding to \eqref{eq:augmented_sys}.

%
%
%

\begin{table}[tb!]
	\centering
	\tabulinesep = 0.2mm
	\caption{RMS tracking errors and weight estimation errors.}
	
	\label{table1}
	\begin{tabu}{lcc}
		\toprule
		
		&\multicolumn{2}{c}{system~1}\\
		&&\\
		&proposed method&comparison alg. \cite{Luo.2016,Kiumarsi.2014}\\
		$\alpha_{\text{RMS}}\!$&$\SI{2.1E-3}{}$&$1.55$\\
		$e_{\text{I}}$&$\SI{6.2E-5}{}$&--\\
		$e_{\text{II}}$&$\SI{2.1E-3}{}$&--\\	
		\bottomrule	
		\\
		\toprule
		
		&\multicolumn{2}{c}{system~2}\\
		&&\\
		&proposed method&comparison alg. \cite{Luo.2016,Kiumarsi.2014}\\
		$y_{\text{RMS}}\!$&$\SI{2.9E-5}{}$&$1.01$\\
		$e_{\text{I}}$&$\SI{2.1E-7}{}$&--\\
		$e_{\text{II}}$&$\SI{1.4E-5}{}$&--\\	
		\bottomrule	
	\end{tabu}
\end{table}

%

\subsection{Results of the Adaptive Tracking Controller}
The RMS tracking errors $\alpha_{\text{RMS}}$ for system~1, $y_{\text{RMS}}$ for system~2 and the corresponding weight estimation errors $e_{\text{I}}$ and $e_{\text{II}}$ are given in Table~\ref{table1}.

Plots of the corresponding tracking performances of our proposed method and the comparison method are given in Fig.~\ref{fig:sys1_x} for system~1 and Fig.~\ref{fig:sys2_x} for system~2. Here, the vertical dash-dotted lines indicate the weight update of our method (i.e. when $k=M$). The reference trajectories $\alpha_{\text{ref}}$ and $y_{\text{ref}}$ are depicted in gray. The black dashed lines show the optimal solutions $\alpha_{\text{opt}}$ and $y_{\text{opt}}$ calculated using full system knowledge and the red line the learned behavior $\alpha_{\text{learned, our}}$ and $y_{\text{learned, our}}$ without knowledge of the system matrices $\m{A}$ and $\m{B}$ when using our proposed method. The resulting tracking behavior $\alpha_{\text{learned, comparison}}$ and $y_{\text{learned, comparison}}$ of the comparison method \cite{Luo.2016,Kiumarsi.2014} is depicted in blue.
\begin{figure}[htb!]
	\begin{center}	
		\newlength\fheight 
		\newlength\fwidth
		\setlength\fheight{7cm} 
		\setlength\fwidth{0.43\textwidth} 
		\input{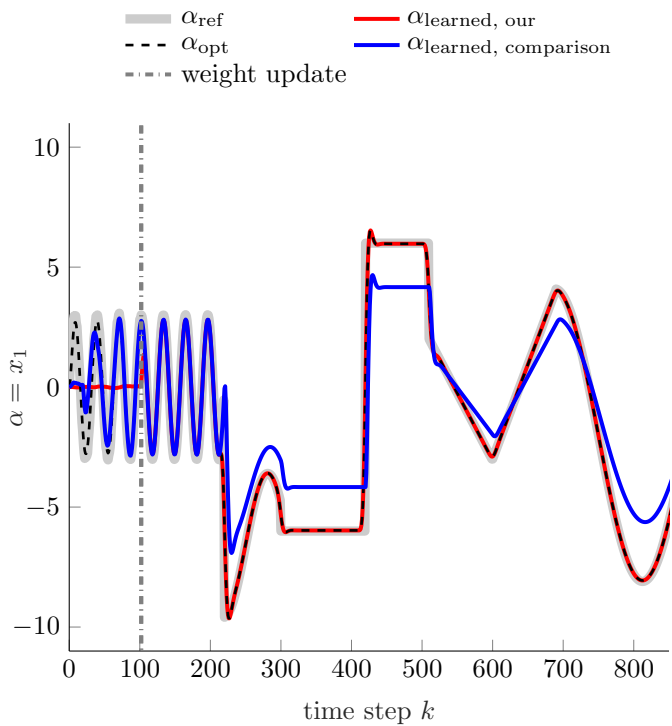}
		\caption{Q-learning based tracking for system~1 (second order rotatory mass-spring-damper model).}\label{fig:sys1_x}	
	\end{center}
\end{figure}
\begin{figure}[htb!]
	\begin{center}	
		\setlength\fheight{7cm} 
		\setlength\fwidth{0.43\textwidth}
		\input{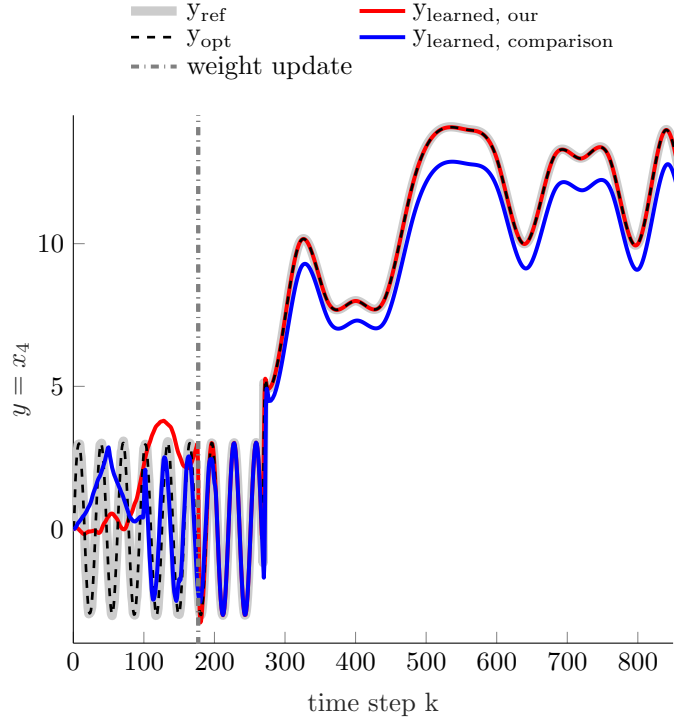}
		\caption{Q-learning based tracking for system~2 (sixth order linear single-track steering model).}\label{fig:sys2_x}	
	\end{center}
\end{figure}

The weight estimation measures $e_{\text{I}}$ and $e_{\text{II}}$ in Table~\ref{table1} indicate that for both systems the optimal Q-function weights have successfully been learned. The decay of $e_{\text{I}}$ and $e_{\text{II}}$ during the value iteration (with iteration index $i$ in Algorithm~\ref{algorithm}) is depicted in Fig.~\ref{fig:sys1_e} (system~1) and Fig.~\ref{fig:sys2_e} (system~2).
\begin{figure}[htb!]
	\begin{center}	
		\setlength\fheight{4cm} 
		\setlength\fwidth{0.43\textwidth} 
%
%
\begin{tikzpicture}
\pgfplotsset{
	compat = 1.8
}
\begin{axis}[%
width=\fwidth,
height=\fheight,
xmin=0,
xmax=30,
xlabel style={font=\color{white!15!black}},
xlabel={iteration $i$},
axis x line*=bottom,
ymin=0,
ymax=1.05,
ytick pos=right,
ylabel={$e_\text{II}$},
ytick={
	0, 0.5, 1
},
line join = round,
yticklabel pos=right,
reverse legend
]
\addplot [color=blue, line width=1.5pt]
  table[row sep=crcr]{%
0	1\\
1	0.999999999991527\\
2	0.935687676615057\\
3	0.751885096048592\\
4	0.483497497089865\\
5	0.234951696739518\\
6	0.0844833696184965\\
7	0.0422934507156081\\
8	0.0321919939727337\\
9	0.0289262049612219\\
10	0.0260230585271198\\
11	0.00839590744626173\\
12	0.0159126538885107\\
13	0.020594426052511\\
14	0.00595664273827021\\
15	0.00352270511995427\\
16	0.00115247619619514\\
17	0.00246915817109376\\
18	0.00285520323014723\\
19	0.00312699003025957\\
20	0.00337273569221423\\
21	0.00279707323403971\\
22	0.00198194563904686\\
23	0.00172059101628957\\
24	0.00204887381106547\\
25	0.00224691542975001\\
26	0.00222132721620205\\
27	0.00217556819843528\\
28	0.00215975340536182\\
29	0.00212747185310962\\
30	0.00209981856484717\\
};\label{plot_one}
\addlegendentry{$e_\text{II}$}
\end{axis}

\begin{axis}[%
width=\fwidth,
height=\fheight,
xmin=0,
xmax=30,
ytick pos=left,
xlabel={iteration $i$},
axis x line*=bottom,
ymin=0,
ymax=0.06,
ylabel={$e_\text{I}$},
line join = round,
reverse legend
]
\addlegendimage{/pgfplots/refstyle=plot_one}\addlegendentry{$e_{\text{II}}$}
\addplot [color=black!20!green, dashed, line width=1.5pt]
table[row sep=crcr]{%
	0	0.0551443768417527\\
	1	0.0509329257868882\\
	2	0.0429836806414018\\
	3	0.0327759385885343\\
	4	0.0226141732122519\\
	5	0.0143081526105289\\
	6	0.00866262221783704\\
	7	0.00544122473698616\\
	8	0.00354103617799151\\
	9	0.00253371799519688\\
	10	0.00167241810317472\\
	11	0.000830156293142492\\
	12	0.000675329463576276\\
	13	0.000543722659439468\\
	14	0.000292873504271034\\
	15	0.000178714916894379\\
	16	0.000109965231524649\\
	17	0.000112295765718014\\
	18	8.58405076180757e-05\\
	19	7.77374474432824e-05\\
	20	8.55293466007243e-05\\
	21	8.50241580641563e-05\\
	22	7.00660217659196e-05\\
	23	5.99103984587873e-05\\
	24	6.09760072811528e-05\\
	25	6.1963646253826e-05\\
	26	6.19144542186293e-05\\
	27	6.28862190253436e-05\\
	28	6.35837585238034e-05\\
	29	6.2904202768688e-05\\
	30	6.19084787536811e-05\\
};
\addlegendentry{$e_\text{I}$}
\end{axis}

\end{tikzpicture}%
		\caption{Decay of weight estimation errors during learning for system~1 (second order rotatory mass-spring-damper model). Here, $e_{\text{I}}$ \eqref{eq:e1} denotes the mean and $e_{\text{II}}$ \eqref{eq:e2} the maximum of the element-wise absolute error of $\m{w}$, both normalized with $\max_j \abs{\left\{\m{w}^*\right\}_j}$.}\label{fig:sys1_e}	
	\end{center}
\end{figure}
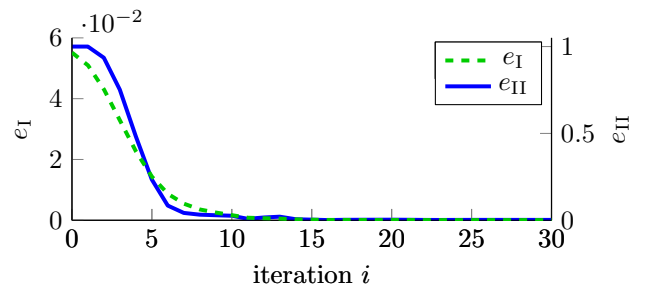
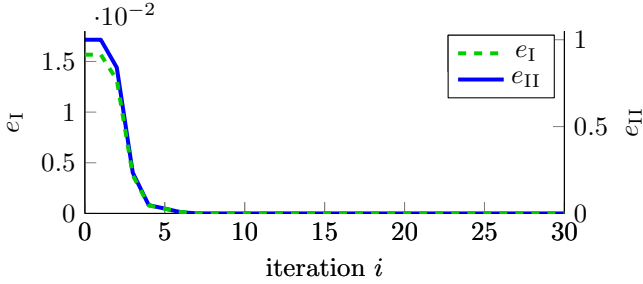
\begin{figure}[htb!]
	\begin{center}	
		\setlength\fheight{4cm} 
		\setlength\fwidth{0.43\textwidth}
%
%
\begin{tikzpicture}
\pgfplotsset{
	compat = 1.8
}
\begin{axis}[%
width=\fwidth,
height=\fheight,
xmin=0,
xmax=30,
xlabel style={font=\color{white!15!black}},
xlabel={iteration $i$},
axis x line*=bottom,
ymin=0,
ymax=1.05,
ytick pos=right,
ylabel={$e_\text{II}$},
ytick={
	0, 0.5, 1
},
line join = round,
yticklabel pos=right,
reverse legend
]
\addplot [color=blue, line width=1.5pt]
table[row sep=crcr]{%
	0	1\\
	1	1.00000000000307\\
	2	0.839704928033393\\
	3	0.234396093151505\\
	4	0.0475996288935857\\
	5	0.0280011713411982\\
	6	0.00720310587225692\\
	7	0.000504985704232911\\
	8	0.000246402735499371\\
	9	6.40093707671484e-05\\
	10	6.26167630249259e-05\\
	11	2.63961421831328e-05\\
	12	0.000286062059373486\\
	13	0.000219298338802716\\
	14	0.000119542101403816\\
	15	0.000152778175405868\\
	16	7.77291861230374e-05\\
	17	4.08194392664941e-05\\
	18	1.22379813286464e-05\\
	19	1.22275556488929e-05\\
	20	1.39303001669059e-05\\
	21	1.35866659757101e-05\\
	22	1.45935623612568e-05\\
	23	1.16131814639323e-05\\
	24	1.37156553182297e-05\\
	25	1.39889730599482e-05\\
	26	1.37975115588641e-05\\
	27	1.39228199646871e-05\\
	28	1.38224079437617e-05\\
	29	1.38330702811662e-05\\
	30	1.38397640518059e-05\\
};\label{plot_one2}
\addlegendentry{$e_\text{II}$}

\end{axis}

\begin{axis}[%
width=\fwidth,
height=\fheight,
xmin=0,
xmax=30,
ytick pos=left,
xlabel={iteration $i$},
axis x line*=bottom,
ymin=0,
ymax=0.018,
ylabel={$e_\text{I}$},
line join = round,
reverse legend
]
\addlegendimage{/pgfplots/refstyle=plot_one2}\addlegendentry{$e_{\text{II}}$}
\addplot [color=black!20!green, dashed, line width=1.5pt]
  table[row sep=crcr]{%
0	0.0156695895762089\\
1	0.015664346712202\\
2	0.0131764085912814\\
3	0.00376648656754866\\
4	0.000787724698570746\\
5	0.000467222308836798\\
6	0.000144314877174705\\
7	2.03761301030951e-05\\
8	1.17347402962018e-05\\
9	3.99536732183135e-06\\
10	3.02030136345071e-06\\
11	6.72135419474038e-07\\
12	4.66363271831048e-06\\
13	4.02486298006763e-06\\
14	2.04760997605958e-06\\
15	2.45955578991422e-06\\
16	1.24184251893618e-06\\
17	7.4584632643349e-07\\
18	3.2196924871787e-07\\
19	2.23762456080346e-07\\
20	2.2427621712278e-07\\
21	2.08480784772708e-07\\
22	2.34718117847757e-07\\
23	1.93197224602682e-07\\
24	2.19374635914016e-07\\
25	2.23572632720569e-07\\
26	2.15088508214826e-07\\
27	2.11848252532001e-07\\
28	2.11653664375243e-07\\
29	2.11798639868587e-07\\
30	2.11696267712447e-07\\
};
\addlegendentry{$e_\text{I}$}
\end{axis}

\end{tikzpicture}%
		\caption{Decay of weight estimation errors during learning for system~2 (sixth order linear single-track steering model). Here, $e_{\text{I}}$ \eqref{eq:e1} denotes the mean and $e_{\text{II}}$ \eqref{eq:e2} the maximum of the element-wise absolute error of $\m{w}$, both normalized with $\max_j \abs{\left\{\m{w}^*\right\}_j}$.}\label{fig:sys2_e}	
	\end{center}
\end{figure}

%

\subsection{Discussion}\label{sec:discussion}
Under the excitation condition \eqref{rankcondition}, our learning controller converges to the optimal control law according to Section~\ref{convergence}. 
Comparing $\alpha_{\text{RMS}}$ (system~1), $y_{\text{RMS}}$ (system~2) and considering Fig.~\ref{fig:sys1_x}, Fig.~\ref{fig:sys2_x} and Fig.~\ref{fig:sys1_x_detail} it is obvious that our algorithm is successfully tracking arbitrary references that have not been seen during training. State of the art methods which assume that the reference follows a time-invariant exo-system $\m{f}_{\text{ref}}$ (e.g. \cite{Luo.2016,Kiumarsi.2014}) are very effective as long as this assumption holds but their tracking performance decreases as soon as the reference to track deviates from the sine described by \eqref{eq:refe}--\eqref{eq:refe1}. This behavior is hardly surprising as these controllers were specifically designed under the assumption of time-invariant exo-system dynamics.

Due to the explicit dependency of our Q-function on the reference on a moving horizon $N$, the learned weights generalize to references that are unknown during the learning procedure such as the ramps, steps and curves in the examples. Due to the optimal behavior according to \eqref{eq:Kostena} our proposed controller exhibits predictive rather than only reactive behavior as can be seen in Fig.~\ref{fig:sys1_x_detail} which depicts a more detailed view of the step in Fig.\ref{fig:sys1_x}.

\begin{figure}[htb!]
	\begin{center}	
		\setlength\fheight{6cm} 
		\setlength\fwidth{0.42\textwidth} 
		\input{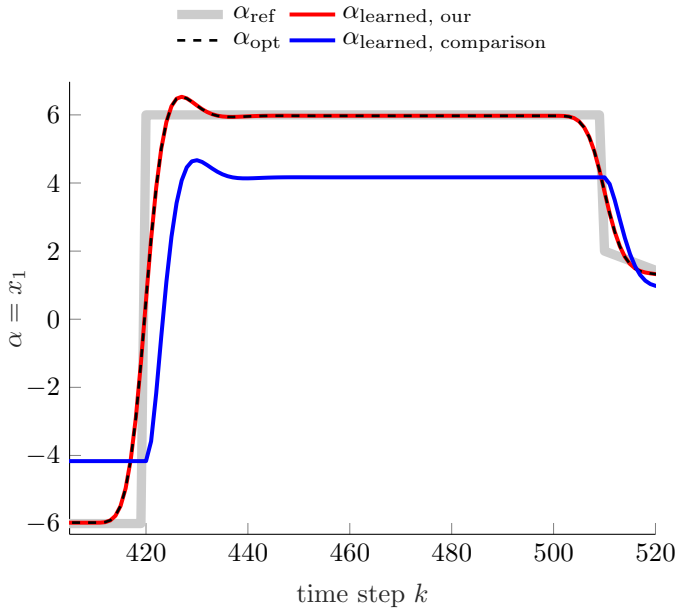}
		\caption{More detailed view of Fig.~\ref{fig:sys1_x} (system~1) to visualize the predictive behavior of our controller due to the moving horizon $N$.}\label{fig:sys1_x_detail}	
	\end{center}
\end{figure}


We would further like to point out that exact knowledge of the structure of $\m{H}_K$ and thus $\m{H}$ which results from Theorem~\ref{Theo:Q} is very beneficial if not vital for the Q-learning method to work efficiently as it helps to reduce the number of weights that have to be learned. If one would only assume $\m{H}$ to be quadratic and symmetric, $L=325$ (system~1) and $L=2701$ (system~2) weights would have to be estimated. Considering Lemma~\ref{lemma:Qexakt_long}, these numbers reduce to $L=247$ (system~1) and $L=2011$ (system~2) and exploiting the sparsity properties of $\m{Q}$ according to Lemma~\ref{lemma:Qexakt}, $L=84$ (system~1) and $L=146$ (system~2) result which renders the complexity tractable.

The choice of the moving horizon length $N$ depends on the available information regarding the reference to track. Furthermore, the larger $N$, the more predictive the learned controller will be but the more unknown weights $\m{w}$ need to be learned (see \eqref{eq:L_sparse}). Thus, an appropriate choice of $N$ obviously depends on the specific application.

\section{Conclusion}\label{conclusion}
In this paper, we proposed a new Reinforcement-Learning-based algorithm that is able to track an arbitrary reference trajectory which is given on a moving horizon while the system dynamics is unknown. In contrast to state-of-the-art methods that are based on RL respectively Adaptive Dynamic Programming, our method does not require the reference trajectory to be generated by an exo-system. It explicitly incorporates arbitrary reference values in a new Q-function. This Q-function, which is constructed such that its minimizing control is part of the solution of the optimal LQ tracking problem, generalizes to arbitrary reference trajectories given on a moving horizon.

We showed that the analytical solution to this Q-function has a specific structure w.r.t. the current state and control as well as the reference on the given horizon. Based thereon, sparsity properties of the resulting structure were exploited in order to reduce the Q-function weights that have to be estimated.
The temporal difference error of the reference-dependent Q-function serves as a target in order to learn the optimal tracking behavior online when the system dynamics is unknown. Here, the choice of basis functions is based on the findings regarding the specific structure of the analytical solution. We proved that this iterative algorithm converges to the optimal solution.

\appendix
\section{Proof of Theorem~\ref{Theo:Q}}\label{append:notations}
Theorem~\ref{Theo:Q} is vital for an efficient Q-function parametrization as it yields the exact structure of the analytical solution of the novel reference-dependent Q-function which reduces the numbers of function approximation weights needed (see Lemma~\ref{lemma:Qexakt_long} and Lemma~\ref{lemma:Qexakt} and the discussion regarding the number of weights to be learned in Section~\ref{sec:discussion}).

Although the notation is complex, it is required to keep the technical proof correct. This is due to the fact that clear distinction of the current time step $k$ and the time step $\kappa$ in the current optimization horizon $K$ starting at $k$ is required. The main idea of this proof is to use dynamic programming and prove the analytical solution of the Q-function by means of backwards induction. Writing down this dynamic programming procedure is cumbersome but recompensates by means of the exact structure of our new Q function.
For $l=1,\dots,d$, the $l$-th submatrix of a matrix $\m{\Pi}\in\mathbb{R}^{p\times nd}$ is defined as
\begin{equation}
\m{\Pi}\left[l\right]=\matc{\m{\Pi}(1,(l-1)n+1) & \cdots & \m{\Pi}(1,nl)\\
	\vdots & \ddots & \vdots\\
	\m{\Pi}(p,(l-1)n+1) & \cdots & \m{\Pi}(p,nl)}.
\end{equation}

Furthermore, note that $\m{I}_n$ denotes the $n\times n$ identity matrix and $\zeta\in\mathbb{N}_0$ is a placeholder. Later, $\zeta$ will be replaced by $\eta$ (respectively $\eta+1$ in the inductive step), where $\eta=K-\kappa$ denotes the remaining time steps on the horizon $K$. Further, $p$ is an index with $p\in\mathbb{N}:$ $p>1$. We define the following shorthand notations.

\begingroup
\setlength{\arraycolsep}{4pt}
\begin{align}
&\begin{aligned}
\m{X}_{\zeta}^{0} =& \m{X}^{0}=\mat{ \m{I}_n & -\m{I}_n },\\
\m{X}_{\zeta}^{1} =& \sqrt{\gamma}\left(-\m{X}^{0}[1]\m{B}\m{G}_\zeta\mat{\m{F}_\zeta & \m{L}_\zeta^{\zeta-1} & \cdots & \m{L}_\zeta^{0}} \right.\\&+\left. \mat{\m{X}^{0}[1]\m{A} & \m{X}^{0}[2] & \m{0} & \cdots & \m{0}}\right)\\
\intertext{and}
\m{X}_{\zeta}^{p} =&		 \sqrt{\gamma}\left(-\m{X}_{\zeta}^{p-1}[1]\m{B}\m{G}_\zeta\mat{\m{F}_\zeta & \m{L}_\zeta^{\zeta-1} & \cdots & \m{L}_\zeta^{0}}\right.\nonumber
\end{aligned}
\\
&\begin{aligned}
+\left. \mat{\m{X}_{\zeta}^{p-1}[1]\m{A} & \m{0} & \m{X}_{\zeta}^{p-1}[2] & \cdots & \m{X}_{\zeta}^{p-1}[\zeta-1]} \right)\!,
\end{aligned}
\end{align}	
\endgroup
as well as
\begingroup
\setlength{\arraycolsep}{4pt}
\begin{align}
&\begin{aligned}
\m{U}_{\zeta}^{1} =& -\m{G}_\zeta\mat{\m{F}_\zeta & \m{L}_\zeta^{\zeta-1} & \cdots & \m{L}_\zeta^{0} }\\
\intertext{and}
\m{U}_{\zeta}^{p} =& \sqrt{\gamma}\left(-\m{U}_{\zeta}^{p-1}[1]\m{B}\m{G}_\zeta\mat{\m{F}_\zeta & \m{L}_\zeta^{\zeta-1} & \cdots & \m{L}_\zeta^{0}}\right.\nonumber
\end{aligned}
\\
&\begin{aligned}
+\left.\mat{\m{U}_{\zeta}^{p-1}[1]\m{A} & \m{0} & \m{U}_{\zeta}^{p-1}[2] & \cdots & \m{U}_{\zeta}^{p-1}[\zeta-1]}\right)\!,
\end{aligned}
\end{align}	
\endgroup
with
\begingroup
\setlength{\arraycolsep}{4pt}
\begin{align}
\m{M}_\zeta &= \gamma\m{B}^\intercal\left(\vphantom{\sum\limits_{i=2}^{\zeta}\left(\m{X}_{K-\zeta}^{i}[1]\right)^\intercal \m{Q} \m{X}_{K-\zeta}^{i}[1]}\right.\sum\limits_{i=0}^{\zeta-2}\left(\m{X}_{\zeta}^{i}[1]\right)^\intercal \m{Q} \m{X}_{\zeta}^{i}[1]\nonumber& \\   &\hphantom{=\gamma\m{B}^\intercal\left(\vphantom{\sum\limits_{i=2}^{\zeta-1}\m{X}[1]_{K-\zeta}^{i^\intercal} \m{Q} \m{X}[1]_{K-\zeta}^{i}}\right.}+\left.\sum\limits_{i=1}^{\zeta-2}\left(\m{U}_{\zeta}^{i}[1]\right)^\intercal \m{R} \m{U}_{\zeta}^{i}[1]\right), &\\
\m{F}_\zeta &= \m{M}_\zeta\m{A},&\\
\m{G}_\zeta^{-1} &= \m{M}_\zeta\m{B}+\m{R}, &\\
\m{L}_\zeta^j &= \begin{cases}\gamma\m{B}^\intercal\!\big(\m{X}^{0}[1]\big)^\intercal \m{Q} \m{X}^{0}[2], \text{ for }j=\zeta-1,&\vspace{0.2cm}\\
\gamma\m{B}^\intercal\!\left(\sum\limits_{i=1}^{\zeta-2}\left(\m{X}_{\zeta}^{i}[1]\right)^\intercal \m{Q} \m{X}_{\zeta}^{i}\left[\zeta-j\right] \right. \\
\hphantom{\gamma\m{B}^\intercal\!\left(\vphantom{\sum\limits_{i=2}^{\zeta-1}\m{U}[1]_{K-\zeta}^{i^\intercal} \m{R}}\right.}  +\left. \sum\limits_{i=1}^{\zeta-2}\left(\m{U}_{\zeta}^{i}[1]\right)^\intercal \m{R} \m{U}_{\zeta}^{i}\left[\zeta-j\right]\right)\!,&\\ \hphantom{\gamma\m{B}^\intercal\m{X}[1]_{K-\zeta}^{j+1^\intercal} \m{Q} \m{X}[2]_{K-\zeta}^{j+1}, \,\quad}\text{for }j<\zeta-1,&\end{cases}
\end{align}
\endgroup
with $j\in\mathbb{N}_0$.

Let furthermore
\begingroup
\setlength{\arraycolsep}{3.3pt}
\begin{align}
\m{\rho}_0^\kappa &= \mat{\m{x}_{k_\kappa}^\intercal & \m{r}_{k_\kappa}^\intercal}\left(\m{X}^{0}\right)^\intercal, \label{eq_rho0}\\
\m{\rho}_1^\kappa &= \mat{\m{x}_{k_\kappa}^\intercal & \m{u}_{k_\kappa}^\intercal & \m{r}_{k_\kappa+1}^\intercal} \matc{\left(\m{X}^{0}[1]\m{A}\right)^\intercal  \\ \left(\m{X}^{0}[1]\m{B}\right)^\intercal \\\left(\m{X}^{0}[2]\right)^\intercal}\!, \\
\m{\mu}_i^\kappa &= \mat{\m{x}_{k_\kappa}^\intercal & \m{u}_{k_\kappa}^\intercal & \m{r}_{k_\kappa+2}^\intercal & \cdots & \m{r}_{k+K}^\intercal}\matc{\left(\m{U}_{\eta}^{\eta-i}[1]\m{A}\right)^\intercal  \\ \left(\m{U}_{\eta}^{\eta-i}[1]\m{B}\right)^\intercal \\\left(\m{U}_{\eta}^{\eta-i}[2]\right)^\intercal \\ \vdots \\ \left(\m{U}_{\eta}^{\eta-i}[\eta]\right)^\intercal}\!, \\
\m{\chi}_i^\kappa &= \mat{\m{x}_{k_\kappa}^\intercal & \m{u}_{k_\kappa}^\intercal & \m{r}_{k_\kappa+2}^\intercal & \cdots & \m{r}_{k+K}^\intercal}\matc{\left(\m{X}_{\eta}^{\eta-i}[1]\m{A}\right)^\intercal  \\ \left(\m{X}_{\eta}^{\eta-i}[1]\m{B}\right)^\intercal \\\left(\m{X}_{\eta}^{\eta-i}[2]\right)^\intercal \\ \vdots \\ \left(\m{X}_{\eta}^{\eta-i}[\eta]\right)^\intercal}\label{eq_xi}\!,
\end{align}
\endgroup
where ${k_\kappa}=k+K-\eta=k+\kappa$ and $i\in\mathbb{N}$.

\begin{proof}
	In the first step, we use backwards induction to prove that the Q-function $\tensor*[^K]{\!Q}{_{\kappa}}$ (cf. Definition~\ref{def:Q}) for system \eqref{eq:System} with the objective function \eqref{eq:Kostena} is given by
	\begin{align}
	\tensor*[^K]{\!Q}{_{\kappa}} &= \frac{1}{2}\!\left(\vphantom{\sum_{i=1}^{N-2}\left(\m{\chi}_i\m{Q}\m{\chi}_i^\intercal+\m{\mu}_i\m{R}\m{\mu}_i^\intercal\right)}
	\m{\rho}_0^\kappa\m{Q}\left(\m{\rho}_0^{\kappa}\right)^\intercal+\m{u}_{k+\kappa}^\intercal \m{R}\m{u}_{k+\kappa}+\gamma\m{\rho}_1^\kappa\m{Q}\left(\m{\rho}_1^{\kappa}\right)^\intercal\right.\nonumber\\ &\hphantom{=\frac{1}{2}} \left.+ \gamma\sum_{i=1}^{K-\kappa-2}\left(\m{\chi}_i^\kappa\m{Q}\left(\m{\chi}_i^{\kappa}\right)^\intercal+\m{\mu}_i^\kappa\m{R}\left(\m{\mu}_i^{\kappa}\right)^\intercal\right)\right).\label{eq:Q_kappa}
	\end{align}
	Starting from $\tensor*[^K]{\!Q}{_{K}}$ (cf. \eqref{eq:QNN}), $\m{u}_{k+K}^*=\m{0}$ directly follows from Definition~\ref{def:Q} and
	\begin{align}
	\left.\frac{\partial \tensor*[^K]{\!Q}{_{K}}}{\partial \m{u}_{k+K}}\right|_{\m{u}_{k+K}^{*}}=\m{0}, \qquad \left.\frac{\partial^2 \tensor*[^K]{\!Q}{_{K}}}{\partial \m{u}_{k+K}^2}\right|_{\m{u}_{k+K}^{*}}=\m{R}\succ\m{0}.
	\end{align}
	Then, by iterating backwards in time, applying \eqref{Q_def1} and the system dynamics \eqref{eq:System}, with $\eta = K-\kappa$, \eqref{eq:Q_kappa} can be shown to hold for $\eta = 0, 1, 2$, i.e. $\kappa = K, K-1, K-2$. Furthermore,
	\begin{align}
	\m{u}_{k+\kappa}^*=-\m{G}_\eta\left(\m{F}_\eta\m{x}_{k+\kappa}+\sum_{j=0}^{\eta-1}\m{L}_\eta^j\m{r}_{k+K-j}\right)\label{eq:uopt_G}
	\end{align}
	minimizes \eqref{eq:Q_kappa} because
	\begin{align}\label{dQdu_d2Qdu2}
	\left.\frac{\partial \tensor*[^K]{\!Q}{_{\kappa}}}{\partial \m{u}_{k+\kappa}}\right|_{\m{u}_{k+\kappa}^{*}}=\m{0},\quad \text{where }
	\left.\frac{\partial^2 \tensor*[^K]{\!Q}{_{\kappa}}}{\partial \m{u}_{k+\kappa}^2}\right|_{\m{u}_{k+\kappa}^{*}}\succ\m{0}
	\end{align}
	is guaranteed as $\m{R}\succ\m{0}$ and $\m{Q}\succeq\m{0}$.
	The induction hypothesis $\tensor*[^K]{\!Q}{_{\kappa-1}}$ (see \eqref{eq:Q_kappa} with $\kappa \rightarrow \kappa-1$) is then proven in the inductive step. This is done by representing $\tensor*[^K]{\!Q}{_{\kappa-1}}$ by means of \eqref{Q_def1} and  utilizing $\m{u}_{k+\kappa}^{*}$ from \eqref{eq:uopt_G}.
	This yields
	\begin{align}\label{eq:kappa-1induction}
	\tensor*[^K]{\!Q}{_{\kappa-1}}&=\frac{1}{2}\mat{\m{x}_{k+\kappa-1}\\ \m{r}_{k+\kappa-1}}^\intercal\left(\m{X}^{0}\right)^\intercal\m{Q}\m{X}^{0}\mat{\m{x}_{k+\kappa-1}\\ \m{r}_{k+\kappa-1}}\nonumber \\
	&\hphantom{ = }+\frac{1}{2}\m{u}_{k+\kappa-1}^\intercal\m{R}\m{u}_{k+\kappa-1}\nonumber \\
	&\hphantom{ = }+\frac{1}{2}\mat{\m{x}_{k+\kappa}\\ \m{r}_{k+\kappa}}^\intercal\left(\m{X}^{0}\right)^\intercal\m{Q}\m{X}^{0}\mat{\m{x}_{k+\kappa}\\ \m{r}_{k+\kappa}}\nonumber \\
	&\hphantom{ = }+\frac{1}{2}\gamma\bar{\m{z}}_k^\intercal \sum_{i=1}^{\eta-1}\left(\left(\m{X}_{\eta+1}^{i}\right)^\intercal\m{Q}\m{X}_{\eta+1}^{i}\right.\nonumber \\ &\left.\hphantom{ = +\frac{1}{2}\gamma\bar{\m{z}}_k^\intercal \sum_{i=1}^{\eta-1}}+\left(\m{U}_{\eta+1}^{i}\right)^\intercal\m{R}\m{U}_{\eta+1}^{i}\right) \bar{\m{z}}_k,
	\end{align}
	where $\bar{\m{z}}_k^\intercal=\mat{\m{x}_{k+\kappa}^\intercal&\m{r}_{k+\kappa+1}^\intercal&\dots&\m{r}_{k+K}^\intercal}$.
	
	Then, replace $\m{x}_{k+\kappa}=\m{A}\m{x}_{k+\kappa-1}+\m{B}\m{u}_{k+\kappa-1}$ (cf. \eqref{eq:System}) in \eqref{eq:kappa-1induction} which results in
	\begin{align}
	\tensor*[^K]{\!Q}{_{{\kappa-1}}} &= \frac{1}{2}\!\left(\vphantom{\sum_{i=1}^{N-2}\left(\m{\chi}_i\m{Q}\m{\chi}_i^\intercal+\m{\mu}_i\m{R}\m{\mu}_i^\intercal\right)}
	\m{\rho}_0^{\kappa-1}\m{Q}\left(\m{\rho}_0^{{\kappa-1}}\right)^\intercal+\m{u}_{k+{\kappa-1}}^\intercal \m{R}\m{u}_{k+{\kappa-1}}\right.\nonumber \\&\hphantom{=\frac{1}{2}X}+\gamma\m{\rho}_1^{\kappa-1}\m{Q}\left(\m{\rho}_1^{{\kappa-1}}\right)^\intercal\nonumber\\ &\hphantom{=\frac{1}{2}X} + \gamma\sum_{i=1}^{K-(\kappa-1)-2}\left(\m{\chi}_i^{\kappa-1}\m{Q}\left(\m{\chi}_i^{{\kappa-1}}\right)^\intercal\nonumber\right. \\ &\hphantom{=\frac{1}{2}X}\left.\left.+\m{\mu}_i^{\kappa-1}\m{R}\left(\m{\mu}_i^{{\kappa-1}}\right)^\intercal\right)\vphantom{\sum_{i=1}^{N-2}\left(\m{\chi}_i\m{Q}\m{\chi}_i^\intercal+\m{\mu}_i\m{R}\m{\mu}_i^\intercal\right)}\right)
	\end{align}
	and yields the induction hypothesis (\eqref{eq:Q_kappa} with $\kappa\rightarrow\kappa-1$) and thus proves \eqref{eq:Q_kappa}.
	
	Thus, the analytical solution of {$\tensor*[^K]{\!Q}{_{0}}$} is quadratic w.r.t. $\m{\rho}_0^\kappa$, $\m{u}_{k}$, $\m{\rho}_1^\kappa$, $\m{\chi}_i^\kappa$ and $\m{\mu}_i^\kappa$. As each of these components is linear w.r.t. $\m{x}_k$, $\m{u}_k$ and $\m{r}_{k+1},\dots,\m{r}_{k+N}$ according to \eqref{eq_rho0}--\eqref{eq_xi}, Theorem~\ref{Theo:Q} follows directly for $\kappa=0$ and $K\geq N$.
\end{proof}

\section{Proof of Lemma~\ref{lemma:VI_matrix_iteration}}\label{append:VI_matrix_iteration}
\begin{proof}
	Let $v(\cdot)$ be a function that transforms a symmetrical squared matrix to a vector such that $v\!\left(\m{H}^{(i)}\right)=\m{w}^{(i)}$. With
	\begin{align}
	&c_k+\gamma{\m{w}^{(i)}}^\intercal\m{\phi}\left(\m{z}_{k+1}^{*(i)}\right)\nonumber\\&=\frac{1}{2}\m{z}_k^\intercal\underbrace{\left(\m{G}+\gamma\m{M}\!\left(\m{L}^{(i)}\right)^\intercal \m{H}^{(i)}\m{M}\!\left(\m{L}^{(i)}\right)\right)}_{=\m{H}^{(i+1)}}\m{z}_k,
	\end{align}
	\eqref{eq:Phimatrix}, \eqref{eq:cmatrix} and the definition of $\m{\phi}$ according to Lemma~\ref{lemma:Qexakt}
	follows that \eqref{eq:LS} can be written as
	\begin{align}
	\m{w}^{(i+1)}&=\left({\m{\Phi}^\intercal}\m{\Phi}\right)^{-1}{\m{\Phi}^\intercal}\matc{\frac{1}{2}\m{z}_{k-M+1}^\intercal\m{H}^{(i+1)}\m{z}_{k-M+1}\\\vdots\\\frac{1}{2}\m{z}_k^\intercal\m{H}^{(i+1)}\m{z}_k}\nonumber \\ &=\underbrace{\left({\m{\Phi}^\intercal}\m{\Phi}\right)^{-1}\left({\m{\Phi}^\intercal}\m{\Phi}\right)}_{=\m{I}}v\left(\m{H}^{(i+1)}\right).\label{eq:lemma_iteration_I}
	\end{align}
	Thus, as $\m{H}^{(i+1)}$ is symmetrically constructed from $\m{w}^{(i+1)}$, it follows from \eqref{eq:lemma_iteration_I} that iterating on $\m{w}^{(i)}$ by means of the value iteration is equivalent to iterating on \eqref{eq:iteration_Hi1}.
\end{proof}

\section{Proof of Lemma~\ref{lemma:monotonically_increasing}}\label{append:monotonically_increasing}
\begin{proof}
	According to Lemma~\ref{lemma:hilfslemma}, the control law $\m{L}\left(\m{Z}^{(i)}\right)$ minimizes $\m{z}_{k+1}^\intercal\m{Z}^{(i)}\m{z}_{k+1}$, $\forall i>0$. Thus,
	\begin{align}
	\m{z}_k^\intercal\m{M}\left(\m{L}\left(\m{W}^{(i)}\right)\right)^\intercal\m{Z}^{(i)}\m{M}\left(\m{L}\left(\m{W}^{(i)}\right)\right)\m{z}_k\nonumber \\
	\geq \m{z}_k^\intercal\m{M}\left(\m{L}\left(\m{Z}^{(i)}\right)\right)^\intercal\m{Z}^{(i)}\m{M}\left(\m{L}\left(\m{Z}^{(i)}\right)\right)\m{z}_k
	\end{align}
	follows. This yields
	\begin{align}
	\m{M}\left(\m{L}\left(\m{W}^{(i)}\right)\right)^\intercal\m{Z}^{(i)}\m{M}\left(\m{L}\left(\m{W}^{(i)}\right)\right)\nonumber \\-\m{M}\left(\m{L}\left(\m{Z}^{(i)}\right)\right)^\intercal\m{Z}^{(i)}\m{M}\left(\m{L}\left(\m{Z}^{(i)}\right)\right)\succeq \m{0}
	\end{align}
	and hence
	\begin{align}\label{eq:Z_Zdach}
	\m{Z}^{(i+1)}=F\left(\m{Z}^{(i)},\m{W}^{(i)}\right)\succeq F\left(\m{Z}^{(i)},\m{L}\left(\m{Z}^{(i)}\right)\right) \eqqcolon \hat{\m{Z}}^{(i+1)}.
	\end{align}
	This also implies
	\begin{align}\label{eq:this_implies}
	F\left(\m{H}^{(i)},\m{L}\left(\m{Z}^{(i)}\right)\right)\succeq F\left(\m{H}^{(i)},\m{L}\left(\m{H}^{(i)}\right)\right) = \m{H}^{(i+1)}.
	\end{align}
	With $\m{0}\preceq\m{H}^{(0)}\preceq\m{Z}^{(0)}$ and the induction hypothesis $\m{H}^{(i)}\preceq\m{Z}^{(i)}$,
	\begin{align}
	&F\left(\m{H}^{(i)},\m{L}\left(\m{Z}^{(i)}\right)\right)\nonumber \\
	&=\m{G}+\gamma\m{M}\left(\m{L}\left(\m{Z}^{(i)}\right)\right)^\intercal\m{H}^{(i)}\m{M}\left(\m{L}\left(\m{Z}^{(i)}\right)\right)\nonumber \\
	&\preceq\m{G}+\gamma\m{M}\left(\m{L}\left(\m{Z}^{(i)}\right)\right)^\intercal\m{Z}^{(i)}\m{M}\left(\m{L}\left(\m{Z}^{(i)}\right)\right)\nonumber \\
	&=F\left(\m{Z}^{(i)},\m{L}\left(\m{Z}^{(i)}\right)\right)=\hat{\m{Z}}^{(i+1)}
	\end{align}
	follows. With \eqref{eq:this_implies}, this yields $\m{H}^{(i+1)}\preceq\hat{\m{Z}}^{(i+1)}$ and incorporating \eqref{eq:Z_Zdach} completes the proof.
\end{proof}

\section{Proof of Lemma~\ref{lemma:upper_bounded}}\label{append:upper_bounded}
\begin{proof}
	Let $\m{Z}^{(0)}=\m{H}^{(0)}$, $\m{Z}^{(i+1)}=F\left(\m{Z}^{(i)},\tilde{\m{L}}\right)$, where $\tilde{\m{L}}$ is chosen such that all eigenvalues of $\left(\m{A}+\m{B}\tilde{\m{L}}_x\right)$ are inside the unit circle. Note that existence of $\tilde{\m{L}}$ is guaranteed due to $\left(\m{A}, \m{B}\right)$ controllable. With $\m{W}^{(i)}=\tilde{\m{L}}$ in Lemma~\ref{lemma:monotonically_increasing}, $\m{0}\preceq\m{H}^{(i)}\preceq\m{Z}^{(i)}$ holds.  With
	\begin{align}
	\m{Z}^{(i+1)}-\m{Z}^{(i)}&=F\left(\m{Z}^{(i)},\tilde{\m{L}}\right)-F\left(\m{Z}^{(i-1)},\tilde{\m{L}}\right)\nonumber \\
	&=\gamma\m{M}\left(\tilde{\m{L}}\right)^\intercal\left(\m{Z}^{(i)}-\m{Z}^{(i-1)}\right)\m{M}\left(\tilde{\m{L}}\right),
	\end{align}
	$\text{vec}({\m{\cdot}})$ stacking the columns of a matrix and $\otimes$ being the Kronecker product,
	\begin{align}
	&\text{vec}\left(\m{Z}^{(i+1)}-\m{Z}{(i)}\right)\nonumber \\&=\underbrace{\gamma\m{M}\left(\tilde{\m{L}}\right)^\intercal\otimes \m{M}\left(\tilde{\m{L}}\right)}_{=\m{E}}\text{vec}\left(\m{Z}^{(i)}-\m{Z}^{(i-1)}\right)
	\end{align}
	follows, thus
	\begin{align}\label{eq:Ehochi1}
	\text{vec}\left(\m{Z}^{(i)}-\m{Z}^{(i-1)}\right)&=\m{E}^{i-1}\text{vec}\left(\m{Z}^{(1)}-\m{Z}^{(0)}\right).
	\end{align}
	If all eigenvalues of $\sqrt{\gamma}\m{M}\left(\tilde{\m{L}}\right)$ are inside the unit circle, this also holds for the eigenvalues of $\m{E}$. Due to its specific structure (cf. \eqref{eq:M}), $(N+1)n$ eigenvalues of $\m{M}\left(\tilde{\m{L}}\right)$ are at the origin. Therefore, consider the remaining eigenvalues, i.e. the eigenvalues of $\m{D}=\matc{\m{A}&\m{B}\\ \tilde{\m{L}}_x\m{A}&\tilde{\m{L}}_x\m{B}}$ (cf. \cite[Lemma~B.1.2]{Landelius.1997}). Let $\norm{\m{\cdot}}$ be the spectral norm of a matrix, respectively the euclidean norm of a vector. Then,
	\begin{align}
	\lim_{i\rightarrow \infty}\norm{\m{D}^i}&=\lim_{i\rightarrow \infty}\norm{\matc{\m{I}_n\\\tilde{\m{L}}_x}\left(\m{A}+\m{B}\tilde{\m{L}}_x\right)^{i-1}\matc{\m{A}&\m{B}}}\nonumber \\
	&\leq \lim_{i\rightarrow \infty}\norm{\matc{\m{I}_n\\\tilde{\m{L}}_x}}\norm{\matc{\m{A}&\m{B}}}\norm{\left(\m{A}+\m{B}\tilde{\m{L}}_x\right)^{i-1}}\nonumber \\ &=0.
	\end{align}
	As $\lim_{i\rightarrow \infty}\m{D}^i=\m{0}$, all eigenvalues of $\m{D}$ are inside the unit circle. Hence, all eigenvalues of $\m{E}$ are also inside the unit circle and $e=\norm{\m{E}}<1$. With
	\begin{align}
	\text{vec}\left(\m{Z}^{(j)}\right)&=\text{vec}\left(\m{Z}^{(0)}\right)+\sum_{i=1}^{j}\text{vec}\left(\m{Z}^{(i)}-\m{Z}^{(i-1)}\right)\nonumber \\ &\overset{\eqref{eq:Ehochi1}}=\text{vec}\left(\m{Z}^{(0)}\right)+\sum_{i=1}^{j}\m{E}^{i-1}\text{vec}\left(\m{Z}^{(1)}-\m{Z}^{(0)}\right)
	\end{align}
	follows
	\begin{align}
	&\norm{\text{vec}\left(\m{Z}^{(j)}\right)}\nonumber \\&\leq\norm{\text{vec}\left(\m{Z}^{(0)}\right)}+\sum_{i=0}^{\infty}\norm{\m{E}}^i\norm{\text{vec}\left(\m{Z}^{(1)}-\m{Z}^{(0)}\right)}=e_0,
	\end{align}
	where the upper bound $e_0$ is independent of $j$.
	As $\norm{\text{vec}\left(\m{Z}^{(j)}\right)}$ is upper bounded by $e_0$, there exists $e_1$ such that $\norm{\m{Z}^{(j)}}\leq e_1$, $\forall j$. With $\m{Y}=e_1\m{I}_{\text{dim}(\m{H})}$, $\m{0}\preceq\m{H}^{(i)}\preceq\m{Z}^{(i)}\preceq\norm{\m{Z}^{(i)}}\m{I}_{\text{dim}(\m{H})}\preceq e_1\m{I}_{\text{dim}(\m{H})}=\m{Y}$ results which completes the proof. 
\end{proof}


 \bibliographystyle{elsarticle-num} 
 \bibliography{mybib} 


%
%
%
\end{document}